\documentclass{article}

\usepackage{microtype}
\usepackage{graphicx}
\usepackage{subfigure}
\usepackage{booktabs} % for professional tables

\usepackage{booktabs}
\usepackage{amsmath,mathrsfs,amssymb,amsfonts}
\usepackage{algorithm} %format of the algorithm
\usepackage{algorithmic} %format of the algorithm
\usepackage{array}
\usepackage{subfigure}
\usepackage{graphicx} 
\usepackage{enumitem}

\usepackage{graphicx} 
\usepackage{bm}
\usepackage{multirow}
\usepackage{bbold}
\usepackage{subfigure}

\usepackage{cases}

\usepackage{hyperref}

\newtheorem{theorem}{Theorem}
\newtheorem{lemma}{Lemma}
\newcommand*{\qed}{\hbox{}$\Box$}
\newenvironment{proof}{\noindent{\it Proof: }}{\qed\medskip}

\newtheorem{remark}[theorem]{Remark}

\usepackage[accepted]{icml2018}

\icmltitlerunning{Unorganized Malicious Attacks Detection}

\begin{document}

\twocolumn[
\icmltitle{Unorganized Malicious Attacks Detection}




\begin{icmlauthorlist}
\icmlauthor{Ming Pang}{}
\icmlauthor{Wei Gao}{}
\icmlauthor{Min Tao}{}
\icmlauthor{Zhi-Hua Zhou}{ed}
\end{icmlauthorlist}

\icmlaffiliation{ed}{National Key Laboratory for Novel Software Technology, Nanjing University, Nanjing 210023, China}

\icmlcorrespondingauthor{Zhi-Hua Zhou}{zhouzh@lamda.nju.edu.cn}

\icmlkeywords{Machine Learning, ICML}

\vskip 0.3in
]



\printAffiliationsAndNotice{}  % leave blank if no need to mention equal contribution

\begin{abstract}
%Recommender system has attracted much attention during the past decade, and many attack detection algorithms have been developed for better recommendation. Most previous approaches focus on the shilling attacks, where the attack organizer fakes a large number of user profiles by the same strategy to promote or demote an item. In this paper, we study a different attack style: \emph{unorganized malicious attacks}, where attackers individually use a small number of user profiles to attack their own target items without any organizer. This attack style occurs in many real applications, yet relevant study remains open. In this paper, we formulate the unorganized malicious attacks detection as a variant of matrix completion problem, and prove that attackers can be detected theoretically. We propose the Unorganized Malicious Attacks detection (UMA) algorithm, which can be viewed as a proximal alternating splitting augmented Lagrangian method. We verify, both theoretically and empirically, the effectiveness of our proposed algorithm.
Recommender system has attracted much attention during the past decade. Many attack detection algorithms have been developed for better recommendations, mostly focusing on shilling attacks, where an attack organizer produces a large number of user profiles by the same strategy to promote or demote an item. This work considers a different attack style: unorganized malicious attacks, where attackers individually utilize a small number of user profiles to attack different items without any organizer. This attack style occurs in many real applications, yet relevant study remains open. We first formulate the unorganized malicious attacks detection as a matrix completion problem, and propose the Unorganized Malicious Attacks detection (\emph{UMA}) approach, a proximal alternating splitting augmented Lagrangian method. We verify, both theoretically and empirically, the effectiveness of our proposed approach.
\end{abstract}

\section{Introduction}
Online activities have been an essential part in our daily life as the flourish of the Internet, e.g., increasing customers prefer shopping on Amazon and eBay; lots of people enjoy watching different movies and TV shows on Youtube and Netflix, etc. There is a big challenge to recommend suitable products effectively as the number of users and items increases drastically; therefore, various collaborative filtering techniques have been developed in diverse systems so as to help customers choose their favorite products in a set of items \citep{li2009transfer,bresler2014latent,rao2015collaborative}.

Many collaborative filtering approaches are vulnerable to spammers and manipulations of ratings \citep{ling2013unified,gunes2014shilling}, and attackers may bias systems by inserting fake rating scores into the user-item rating matrix. Some attackers may increase the popularity of their own items (push attack) while some others may decrease the popularity of their competitors' items (nuke attack). Most attack detection studies focus on shilling attacks and show good detection performance on kinds of shilling attack strategies \citep{mehta2007unsupervised,hurley2009statistical,ling2013unified}. They consider the situation that all the attack profiles are produced by the same strategy to promote or demote a particular item. For example, an attack organizer may produce hundreds of fake user profiles by a strategy that each fake user profile gives high scores to the most popular movies and low scores to the target movie to demote it.

Various practical techniques have been developed to control shilling attacks, e.g., online sites require real names and telephone numbers for registrations; CAPTCHA is used to determine that the response is not generated by a robot; customers are allowed to rate a product after purchasing this product on the shopping website. Based on these measures, traditional shilling attacks may suffer high cost. For example, small online sellers in e-commerce like Amazon might not be willing to produce hundreds of fake customer rating profiles to implement a shilling attack.

In this paper, we investigate a new attack model named \emph{unorganized malicious attack}, where attackers individually use a small number of user profiles to attack their own targets without any organizer. This attack style happens in many real applications, e.g., online sellers on Amazon may fake a few customer rating profiles to demote their competitors' high-quality shoes; writers may hire several readers to give high scores to promote their low-quality books. In fact, it has been shown that systems are seriously affected by small amounts of unorganized malicious attacks, e.g., the first maliciously bad rating can decrease the sales of one seller by 13\%~\cite{luca2016reviews}.

We first formulate the unorganized malicious attacks detection as a variant of matrix completion problem. Let $X$ denote the ground-truth rating matrix without attacks and noises, and the matrix is low-rank since the users' preferences are affected by several factors~\citep{salakhutdinov2007restricted}. Let $Y$ be the sparse attack-score matrix, and $Z$ denotes a noisy matrix. What we can observe is a (or partial) matrix $M$ such that $M=X+Y+Z$. As far as we know, previous works do not make similar formulation for attack detection. The main difference between our optimization problem and robust PCA \citep{candes2011robust} is that roubst PCA focuses on recovering low-rank part $X$ from complete or incomplete matrix and we pay more attention to distinguishing the sparse attack term $Y$ from the small perturbation noise term $Z$.

Theoretically, we prove that the low-rank rating matrix $X$ and the sparse matrix $Y$ can be recovered under some classical matrix-completion assumptions. We propose the Unorganized Malicious Attacks detection (UMA) algorithm, which can be viewed as a proximal alternating splitting augmented Lagrangian method. Some new techniques have been developed to prove its global convergence with a worst-case $O(1/t)$ convergence rate. Finally, experimental results verify the effectiveness of our proposed algorithm in comparison with the state-of-the-art methods of attack detection.

The rest of this paper is organized as follows. Section~2 reviews some related works. Section~3 introduces the framework of unorganized malicious attacks detection, and Section~4 proposes our UMA algorithm. Section~5 presents theoretical justification for matrix recovery and convergence of UMA. Section~6 shows our experiments, and Section~7 concludes this work.

\section{Related Work}
Collaborative filtering (CF) is one of the most successful techniques to build recommender systems. The core assumption of CF is that if users express similar interests in the past, they will share common interest in the future \citep{goldberg1992using}. Significant progress about CF has been made since then \citep{salakhutdinov2007restricted,li2009transfer,bresler2014latent,rao2015collaborative}. There are two main categories of conventional CF (based on the user-item rating matrix) which are memory-based and model-based CF algorithms. Memory-based CF predicts a user's rating on an item based on the entire or part of the user-item matrix. It can be subdivided into user-based and item-based CF. A typical user-based CF approach predicts the ratings of a user by aggregating the ratings of some similar users. User similarity is defined by a similarity metric, usually the cosine similarity or the Pearson correlation \citep{singhal2001modern}. Many modifications and adjustments about the similarity metric have been proposed \citep{adomavicius2005toward,zhang2007recursive}. Item-based CF approaches predict the rating of an item for a user according to the ratings of items the user has given \citep{deshpande2004item}.

Model-based CF approaches use the user-item matrix to train prediction models and recommendations are generated from the prediction models \citep{ekstrand2011collaborative}. For example, the mixture model learns the probability distribution of items in each clusters \citep{kleinberg2008using}; Matrix factorization techniques learn latent factors of users and items from the user-item matrix and then use the low-rank approximation matrix to predict the score of unrated items; From probabilistic perspective, \citet{salakhutdinov2008bayesian} propose probabilistic matrix factorization framework. Considering about side information besides the user-item matrix, many works expand the CF paradigm \citep{basilico2004unifying,salakhutdinov2007restricted}.

However the two main categories of CF schemes are both vulnerable to attacks \citep{gunes2014shilling,aggarwal2016recommender}. Increasing attention has been given to attack detection. Researchers have proposed several kinds of methods which can be mainly thought as statistical, clustering, classification and data reduction-based methods \citep{gunes2014shilling}. These methods mainly focus on shilling attacks where the attack organizer produces a large number of user profiles by the same strategy to promote or demote a particular item. Statistical methods are used to detect anomalies who give suspicious ratings. \citet{hurley2009statistical} propose a Neyman-Pearson statistical attack detection method to distinguish attackers from normal users. Similarly, probabilistic Bayesian network models are used in \citet{li2011detection}. Based on attributes derived from user profiles, classification methods detect attacks by kNN, SVM, rough set theory, etc~ \citep{mobasher2007attacks,he2010attack}. 

An unsupervised clustering algorithm based on several classification attributes \citep{bryan2008unsupervised} is presented in~\citet{bhaumik2011clustering}. They apply $k$-means clustering based on these attributes and classify users in the smallest cluster as attackers. Instead of using traditional nearest neighbor methods, \citet{mehta2007unsupervised} proposes a PLSA-based clustering method. \citet{mehta2009unsupervised} propose the variable selection method, which treats users as variables and calculates their covariance matrix. By analyzing the principal components of the covariance matrix, those users with the smallest coefficient in the first $l$ principal components are chosen in the final variable selection. \citet{ling2013unified} try to use a low-rank matrix factorization method to predict the users' ratings. Users' reputation is computed according to the predicted ratings and low-reputed users are classified as malicious users. 

These methods implement detection based on the common characteristics of the attack profiles produced by the same attack strategy. When recommender systems are under unorganized malicious attacks, different attackers use different strategies to produce attack profiles or hire existing users to attack different targets. The traditional attack detection methods may be not suitable in this case.

Recovering low-dimensioinal structrue from a corrupted matrix is related to robust PCA~\cite{candes2011robust,yi2016fast,bouwmans2017decomposition}. They focus on recovering low-rank part $X$ from complete or incomplete matrix, and it is different from attacks detection (which is our task). We pay more attention to distinguishing the sparse attack term $Y$ from the small perturbation noise term $Z$.

%RPCA and Huan's work.... (according to the format of Huan's last para of Intro), add more references

\section{The Formulation}
In this section, we introduce the general form of an attack profile, and then we give a detailed comparison of unorganized malicious attacks and shilling attacks. The formal definition of unorganized malicious attacks and the corresponding detection problem formulation are also presented.

\subsection{Notations}
We begin with some notations used throughout this paper. Let $\|X\|$, $\|X\|_F$ and $\|X\|_*$ denote the operator norm, Frobenius norm and nuclear norm of matrix $X$, respectively. Let $\|X\|_1$ and $\|X\|_\infty$ be the $\ell_1$ and $\ell_\infty$ norm of matrix $X$ seen as a long vector, respectively. Further, we define the Euclidean inner product between two matrices as $\langle X, Y\rangle:=\text{trace}(XY^\top)$, where $Y^\top$ means the transposition of $Y$. Therefore, We have $\|X\|_F^2=\langle X,X\rangle$.

Let $P_\Omega$ denote an operator of linear transformation which acts on matrices space, and we also denote $P_\Omega$ by the linear space of matrices supported on $\Omega$ when it is clear from the context. Then, $P_{\Omega^\top}$ represents the space of matrices supported on $\Omega^c$. For an integer $m$, let $[m]:=\{1,2,\ldots,m\}$.

\begin{figure}[ht]
	%\vskip 0.2in
	\begin{center}
		%\vspace{-1in}
		%\centerline{\includegraphics[width=\columnwidth]{figure3/generalAttack.pdf}}
		\centerline{\includegraphics[width=\columnwidth]{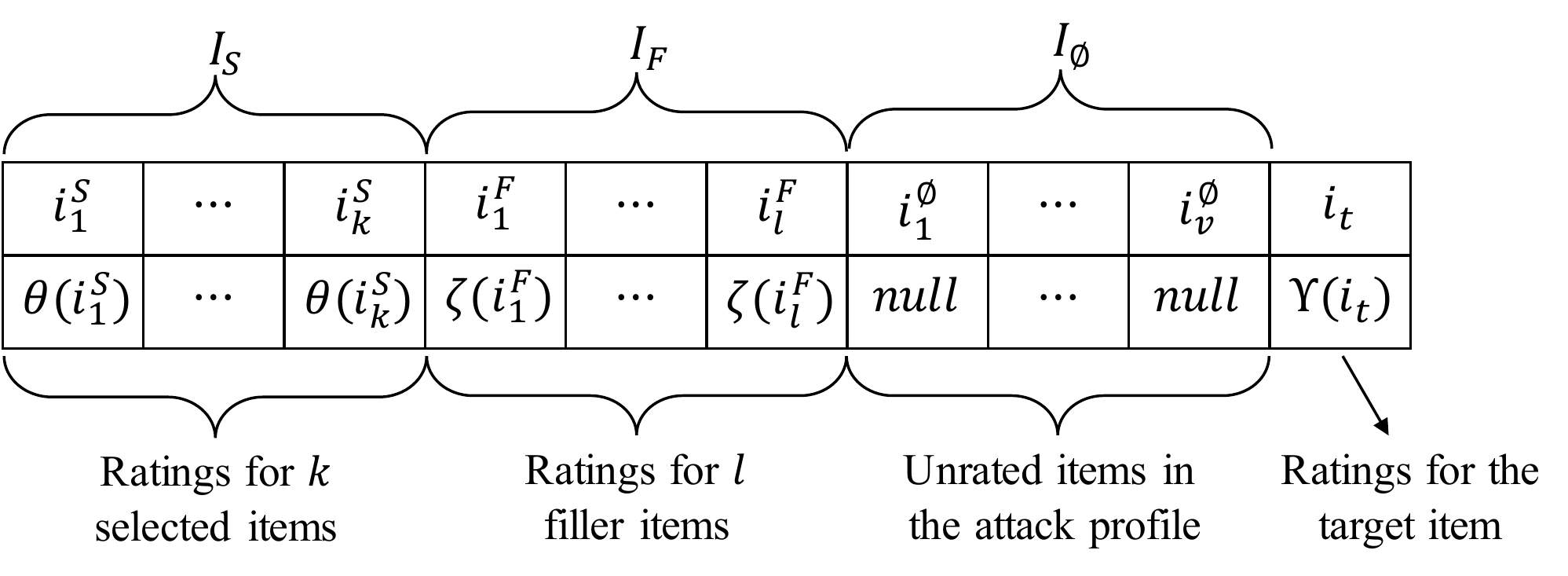}}
		%\vspace{-0.8in}
		\caption{General form of an attack profile.}
		\label{generalAttack}
	\end{center}
	%\vskip -0.2in
\end{figure}

\subsection{Problem Formulation}\label{problem-form}
The general form of an attack profile is shown in Figure~\ref{generalAttack} which is first defined by \citet{bhaumik2006securing}. The attack profile is partitioned in four parts. The null part, $I_{\emptyset}$, contains the items with no ratings. The single target item $i_t$ will be given a malicious rating, i.e. a high rating in a push attack and a low rating in a nuke attack. The selected items, $I_S$, are a group of selected items for special treatment during the attack. The filler items, $I_F$, are selected randomly to complete the attack profile. Functions $\theta$, $\zeta$ and $\Upsilon$ determine how ratings should be assigned to items in $I_S$, $I_F$ and target item $i_t$. Three basic attack strategies are listed bellow.
\begin{itemize}
	\item Random attack: $I_S$ is empty, and $I_F$ is selected randomly with function $\zeta$ generating random ratings centered around the overall average rating in the database.
	\item Average attack: $I_S$ is empty, and $I_F$ is selected randomly with function $\zeta$ generating random ratings centered around the average rating for each item.
	\item Bandwagon attack: $I_S$ is selected from the popular items and  function $\theta$ gives $I_S$ high ratings. The filler items are handled similarly to random attack.
\end{itemize}

In each shilling attack, one attack strategy is chosen, e.g., average attack strategy; the target item $i_t$ is fixed; the numbers of rated items $k$ and $l$ are fixed; the rating functions are fixed.
This leads to the generated attack profiles having some common characteristics in one shilling attack. Besides, a large number of attack profiles is required in the basic settings of shilling attacks.

However, in unorganized malicious attacks, the attack strategies, the number of rated items, the target item and the rating functions are not constrained to be the same. Besides, we assume that each attacker produces a small number of attack profiles. The formal definition of unorganized malicious attacks is given in the following.

Let $U_{[m]}=\{U_1,U_2,\ldots,U_m\}$ and $I_{[n]}=\{I_1,I_2,\ldots,I_n\}$ denote $m$ users and $n$ items, respectively. Let $X\in{\cal R}^{m\times n}$ be the ground-truth rating matrix. $X_{ij}$ denotes the score that user $U_i$ gives item $I_j$ without any attack or noise, i.e., $X_{ij}$ reflects the ground-truth feeling of user $U_i$ on item $I_j$. Suppose that the score range is $[-R,R]$, then $-R\le X\le R$. Throughout this work, we assume that $X$ is a low-rank matrix as in classical matrix completion \citep{salakhutdinov2008probabilistic}. The intuition behind this assumption is that only a few factors influence users' preferences. 

Usually, the ground-truth matrix $X$ may be corrupted by a system noisy matrix $Z$. For example, if $X_{ij}=4.8$ for $i\in[m]$, then, it is acceptable that user $U_i$ gives item $I_j$ score $5$ or $4.6$. In this paper, we consider the independent Gaussian noise, i.e., $Z=(Z_{ij})_{m\times n}$ where each element $Z_{ij}$ is drawn i.i.d. from the Gaussian distribution $\mathcal{N}(0,\sigma)$ with parameter $\sigma$.

Let $M$ be the observed rating matrix. We define unorganized malicious attacks as follows: for every $j\in [n]$, $|U^j|<\gamma$ where $U^j=\{U_i|i\in [m]\; \&\; |M_{ij}-X_{ij}|\ge \epsilon\}$. The parameter $\epsilon$ is used to distinguish malicious users from the normal, and the parameter $\gamma$ limits the number of user profiles attacking one item. Intuitively, unorganized malicious attacks mean that attackers individually use a small number of user profiles to attack their own targets. So multiple independent shilling attacks where each contains a small number of attack profiles can be regarded as an example of unorganized malicious attacks.

%Let $M$ be the observed matrix which is corrupted by random noise and unorganized malicious attacks. We define an unorganized malicious attack with respect to a user set $U_K$ ($K\subset [m]$) if $|K|\le \iota$ and $|M_{ij}-X_{ij}|\ge \epsilon$ for some $j\in[n]$ and every $i\in K$. The parameter $\iota$ controls the number of users and parameter $\epsilon$ is used to distinguish malicious users from the normal. Intuitively, unorganized malicious attacks mean that attackers individually use a small number of user profiles to attack their own targets.

It is necessary to distinguish unorganized malicious attacks from noise. Generally speaking, user $U_i$ gives item $I_j$ a normal score if $|M_{ij}-X_{ij}|$ is very small, while user $U_i$ makes an attack to item $I_j$ if $|M_{ij}-X_{ij}|\ge \epsilon$. For example, if the ground-truth score of item $I_j$ is $4.8$ for user $U_i$, then user $U_i$ makes a noisy rating by giving $I_j$ score $5$, yet makes an attack by giving $I_j$ score $-3$. Therefore, we assume that $\|Z\|_F\leq \delta$, where $\delta$ is a small parameter.

Let $Y=M-X-Z=(Y_{ij})_{m\times n}$ be the malicious-attack-score matrix. Then, $Y_{ij}=0$ if user $U_i$ does not attack item $I_j$, otherwise $|Y_{ij}|\geq \epsilon$. We assume that $Y$ is a sparse matrix, and the intuition behind this assumption is that the ratio of malicious ratings to all the ratings is small. Notice that we can not directly recover $X$ and $Y$ from $M$ because such recovery is an NP-Hard problem \citep{candes2011robust}. We consider the following optimization problem.
\begin{equation}\label{equ:originForm}
\begin{split}
\min_{X,Y,Z}\,&\|X\|_*+\tau\|Y\|_1-\alpha \langle M,Y\rangle+\frac{\kappa}{2}\|Y\|^2_F\\
\text{s.t. }&X+Y+Z=M, \\
&\|Z\|_F\leq \delta.
\end{split}
\end{equation}
$\|X\|_*$ acts as a convex surrogate of the rank function. $\|Y\|_1$ is used to induce the sparse attack component of $M$. The term $\langle M,Y\rangle$ is introduced to better distinguish $Y$ and $Z$. While each entry in $Z$ is small and $Z_{ij}M_{ij}$ can be either positive or negative, the malicious rating bias $Y_{ij}$ and the observed rating $M_{ij}$ are usually the same as positive or negative, i.e. $M_{ij}Y_{ij}> 0$. Specifically, if it is a nuke attack, $Y_{ij}<0$ and $M_{ij}<0$; if it is a push attack, $Y_{ij}>0$ and $M_{ij}>0$. $\|Y\|^2_F$ is another strongly convex regularizer for $Y$. This term also guarantee that we can get the optimal solution efficiently.
$\tau$, $\alpha$ and $\kappa$ are tradeoff parameters.

%where the term $\langle M,Y\rangle$ is introduced to control the impact of $Y$ in the nuke attack detection. The intuition behind this term is that the malicious rating bias $Y_{ij}$ and the true rating $M_{ij}$ are always opposite in nuke attack, i.e. $X_{ij}Y_{ij}\leq 0$. So we need to minimize $\langle M,Y\rangle$ to better distinguish sparse matrix $Y$ and noisy matrix $Z$. The third term is $\alpha \langle (M-R_{max}),Y\rangle$ in the push attack detection.

In many real applications, we can not get a full matrix $M$, and only partial entries can be observed. Let $\Omega\in[m]\times[n]$ be the set of observed entries. We define an orthogonal projection $P_{\Omega}$ onto the linear space of matrices supported on $\Omega \subset [m]\times[n]$ as follows.
\begin{equation*}
P_{\Omega}M=\left\{
\begin{array}{ll}
M_{ij}& \text{ for } (i,j)\in\Omega,\\
0& \text{ otherwise.} \\
\end{array}
\right.
\end{equation*}
Our final optimization framework for unorganized malicious attack detection can be formulated as follows.
\begin{equation}\label{equ:originForm2}
\begin{split}
\min_{X,Y,Z}\,&\|X\|_*+\tau\|Y\|_1-\alpha \langle M,Y\rangle+\frac{\kappa}{2}\|Y\|^2_F\\
\text{s.t. }&P_{\Omega}(X+Y+Z)=P_{\Omega} M,\\ &\|P_{\Omega}(Z)\|_F\leq\delta.
\end{split}
\end{equation}

There have been many works focus on recovering low-rank part $X$ from complete or incomplete matrix \citep{candes2011robust, mackey2011divide, peng2012rasl, feng2013online}. However, we pay more attention to distinguishing the sparse noise term $Y$ from the small perturbation term $Z$. In order to find nonzero entries of $Y$, a new term $\langle M,Y\rangle$ is added which 
%leads to a more challenging optimization task and 
gets a much better performance. Further details about the proposed approach, theoretical analysis and experiments are explained below in Section~\ref{proposed_approach}, \ref{theoretical_analysis}, \ref{experiments}.

\section{The Proposed Approach}\label{proposed_approach}

In this section, we propose a proximal alternating splitting augmented Lagrangian method to solve the optimization problem~\eqref{equ:originForm2}, which can guarantee global convergence.

The separable structure emerging in both the objective function and constrains in Eq.~\eqref{equ:originForm2} enables us to derive an efficient algorithm by splitting the optimization problem in appropriate ways.
However, it is rather difficult to optimize this problem with theoretical guarantee, because this optimization involves three blocks variables. The existing literatures have shown that the direct extension of the alternating direction method of multipliers may not be convergent for solving  Eq.~\eqref{equ:originForm2}, e.g., \cite{chen2016direct}.

%In this section, we reformulate (\ref{equ:originForm2}) into an augmented-Lagrangian-oriented form which is beneficial for the algorithmic design and convergence analysis later on.

%It is rather difficult to optimize the formulations of Eqns.~\eqref{equ:originForm} and \eqref{equ:originForm2} with theoretical guarantees, because this optimization includes three-block non-convex programming with coupled objective function $\alpha\langle X,Y\rangle$. The difficulties lie in three hands: i) the objective function is non-convex due to the existence of coupled term; ii) the subgradients of the involved three functions are non-Lipschitz; and iii) there are three blocks variables involved. Therefore, we consider the following perturbation formulation inspired by \citet{cai2010singular,candes2011robust}.

Firstly, we reformulate Eq.~\eqref{equ:originForm2} as follows,
\begin{eqnarray}\label{equ:finalForm}
	\begin{array}{cl}
		\min&\|X\|_*+\tau\|Y\|_1-\alpha\langle M,Y\rangle+\frac{\kappa}{2}\|Y\|_F^2\\
		s.t.&X+Y+Z=P_{\Omega} M,\\
		&Z\in{\bf B}, \\
		& {\bf B}:=\{Z|\|P_{\Omega}(Z)\|_F\le \delta\},
	\end{array}
\end{eqnarray}
where $\kappa>0$ is a regularization parameter. As $\kappa\to 0$ and $\alpha\to 0$, this formulation degenerates into robust PCA considering small dense noise. We further get the augmented Lagrangian function as follows.
\begin{eqnarray}\label{equ:lagrangianForm}
	&\mathcal{L}_{\mathcal{A}}(X,Y,Z,\Lambda,\beta):=\|X\|_{*}+\tau \|Y\|_1-\alpha\langle M,Y\rangle&\nonumber\\ 
	&+\frac{\kappa}{2}\|Y\|_F^2-\langle \Lambda,L\rangle+\frac{\beta}{2}\| L \|_F^2&
\end{eqnarray}
where $L=X+Y+Z- P_{\Omega} M$ and $\beta$ is a positive constant.

It is noteworthy that traditional algorithms, e.g., \citep{tao2011recovering,he2015splitting} can not guarantee the convergence for this three-block convex minimization problem. 
We propose a proximal alternating splitting augmented Lagrangian method to decompose the optimization task of Eq.~\eqref{equ:finalForm} into three smaller ones which solve the variables $Z^{k+1}$, $X^{k+1}$ and $Y^{k+1}$ separately in the consecutive order. We will provide global convergence guarantee and convergence rate for this algorithm in Section~\ref{theoretical_analysis}.

%It is noteworthy that there is a coupling term in our objective function, and some traditional algorithms \citep{tao2011recovering,he2015splitting} can not be applied directly. We propose a proximal alternating splitting augmented Lagrangian method to solve this optimization, which inherits the advantages of ASALM \citep{tao2011recovering}. We will provide global convergence guarantee for this algorithm in Section~\ref{sec:convergence}.

More specifically, let $\mathcal{L}_{\mathcal{A}}$ and $\bf B$ be defined in Eq.~\eqref{equ:finalForm}. With given $(X^{k},Y^{k},\Lambda^k)$, we first update
\[
Z^{k+1} = \arg\min_{Z \in {\bf B}}   \mathcal{L}_{\mathcal{A}}(X^k,Y^{k},Z,\Lambda^k,\beta),
\]
and it is easy to get the closed form solution
\begin{equation}\label{eqnRefreshZ}
Z_{ij}^{k+1}=\left\{
\begin{array}{ll}
\min\{1, \delta/\|P_{\Omega}N\|_F\}N_{ij} &\text{if }(i,j)\in\Omega  \\
N_{ij} & \text{otherwise}
\end{array}
\right.
\end{equation}
where $N=\frac{1}{\beta}\Lambda^k+P_\Omega M-X^k-Y^k$. Then, we update
\begin{equation*}
X^{k+1} = \arg\min_{X\in {\cal R}^{m \times n}}  \mathcal{L}_{\mathcal{A}}(X,Y^k,Z^{k+1},\Lambda^{k},\beta),
\end{equation*}
and Lemma~\ref{lem:t2} gives the closed solution as
\begin{equation}\label{eqnRefreshX}
X^{k+1}=\mathcal{D}_{1/\beta}(P_\Omega M+\frac{1}{\beta}\Lambda^k-Y^{k}-Z^{k+1})
\end{equation}
where the nuclear-norm-involved shrinkage operator $\mathcal{D}_{1/\beta}$ is defined in Lemma~\ref{lem:t2}. Further, we update
\[
Y^{k+1} = \arg\min_{Y \in {\cal R}^{m \times n}}  \mathcal{L}_{\mathcal{A}}(X^{k+1},Y,Z^{k+1},\Lambda^{k},\beta)
\]
and Lemma~\ref{lem:t1} gives the closed solution $Y^{k+1}$ as
\begin{equation}\label{eqnRefreshY}
\mathcal{S}_{\tau/\beta}(\frac{\alpha+\beta}{\beta}P_\Omega M+\frac{1}{\beta}\Lambda^k-Z^{k+1}-X^{k+1})\upsilon\beta
\end{equation}
%\begin{equation}\label{eqnRefreshY}
%\mathcal{S}_{\tau\upsilon}\big(\upsilon\beta(\frac{\alpha+\beta}{\beta}P_\Omega %M+\frac{1}{\beta}\Lambda^k-Z^{k+1}-X^{k+1})\big)
%\end{equation}
where $\upsilon=1/(\beta+\kappa)$ and the shrinkage operator $\mathcal{S}_{\tau/\beta}$ is defined in Lemma~\ref{lem:t1}. Finally, we update
\[
\Lambda^{k+1} = \Lambda^k - \beta(X^{k+1} + Y^{k+1} +Z^{k+1} - P_\Omega M).
\]
The pseudocode of the proposed UMA algorithm is given in Algorithm~\ref{alg:UMA}.

\begin{algorithm}[tb]
	%\floatname{algorithm}{The UMA Algorithm}
	\caption{The UMA Algorithm}
	\label{alg:UMA}
	\textbf{Input:} matrix $M$ and parameters $\tau$, $\alpha$, $\beta$, $\delta$ and $\kappa$.\\
	\textbf{Output:} Label vector $[y_1,\dots,y_m]$ where $y_i=1$ if user $U_i$ is a malicious user; otherwise $y_i=0$.\\
	\textbf{Initialize:} $Y^0=X^0=\Lambda^0=0$, $y_i=0$ ($i=1,\dots,m$), $k=0$\\
	\textbf{Process:}
	\begin{algorithmic}[1]
		\WHILE{not converged}
		\STATE Compute $Z^{k+1}$ by Eq.~\eqref{eqnRefreshZ}.
		\STATE Compute $X^{k+1}$ by Eq.~\eqref{eqnRefreshX}.
		\STATE Compute $Y^{k+1}$ by Eq.~\eqref{eqnRefreshY}.
		\STATE Update the Lagrange multiplier $\Lambda^{k+1}$ by\\ $\quad\quad \Lambda^k-\beta(X^{k+1}+Y^{k+1}+Z^{k+1}-P_\Omega M)$.
		\STATE $k=k+1$.
		\ENDWHILE
		\STATE if $\max(|Y_{i,:}|)>0$, then $y_i=1$ ($i=1,\dots,m$).
	\end{algorithmic}
\end{algorithm}

\section{Theoretical Analysis}\label{theoretical_analysis}
%In this section, we first prove that attacks can be detected theoreteically. Then we give a convergence analysis of our algorithm UMA and prove its global convergence.
This section presents our main theoretical results. Due to the page limit, the detailed proofs and analysis are given in the supplement document.

%\subsection{Detection Guarantees}
We begin with two useful lemmas for the deviation of our proposed algorithm as follows.
\begin{lemma}\citep{bruckstein2009sparse}\label{lem:t1}
	For $\tau>0$ and $T\in {\cal R}^{m\times n}$, the closed solution of $\min_{Y}\tau\|Y\|_1+\|Y-T\|_F^2/2$ is given by matrix $\mathcal{S}_\tau (T)$ with
	$(\mathcal{S}_{\tau}(T))_{ij}=max\{|T_{ij}|-\tau ,0\}\cdot\text{sgn}(T_{ij})$,
	where $\text{sgn}(\cdot)$ means the sign function.
\end{lemma}
\begin{lemma}\citep{cai2010singular}\label{lem:t2}
	For $\mu>0$ and $Y\in{\cal R}^{m\times n}$ with rank $r$, the closed solution of $\min_{X}\mu\|X\|_*+\|X-Y\|_F^2/2$ is given by
	\[
	\mathcal{D}_\mu(Y)=S\text{diag} (\mathcal{S}_\mu(\Sigma))D^\top
	\]
	where $Y=S\Sigma D^\top$ denotes the singular value decomposition of $Y$, and $\mathcal{S}_\mu(\Sigma)$ is defined in Lemma~\ref{lem:t1}.
\end{lemma}

We now present theoretical guarantee that UMA can recover the low-rank component $X$ and the sparse component $Y$. 
For simplicity, our theoretical analysis focuses on square matrix, and it is natural to generalize our results to the general rectangular matrices. 

Let the singular value decomposition of $X_0\in{\cal R}^{n\times n}$ be given by
\[
X_0=S\Sigma D^\top=\sum\nolimits_{i=1}^r \sigma_is_id_i^\top
\]
where $r$ is the rank of matrix $X_0$, $\sigma_1,\ldots,\sigma_r$ are the positive singular values, and $S=[s_1,\ldots,s_r]$ and $D=[d_1,\ldots,d_r]$ are the left- and right-singular matrices, respectively. For $\mu>0$, we assume
\begin{eqnarray}\label{eq:incoherence}
\max_i\|S^\top e_i\|^2&\leq& {\mu r}/{n},\nonumber\\
\max_i\|D^\top e_i\|^2&\leq& {\mu r}/{n},\\
\|SD^\top\|_\infty^2&\leq& \mu r/n^2.\nonumber
\end{eqnarray}
%We now present our first main result as follows.

\begin{theorem}\label{thm1}

Suppose that the support set of $Y_0$ be uniformly distributed for all sets of cardinality $k$,  and $X_0$ satisfies the incoherence condition given by Eq.~\eqref{eq:incoherence}. Let $X$ and $Y$ be the solution of optimization problem given by Eq.~\eqref{equ:originForm} with parameter $\tau=O(1/\sqrt{n})$, $\kappa=O(1/\sqrt{n})$ and $\alpha=O(1/n)$. For some constant $c>0$ and sufficiently large $n$, the following holds with probability at least $1-cn^{-10}$ over the choice on the support of $Y_0$,

\[
\|X_0-X\|\leq \delta \text{ and }\|Y_0-Y\|_F\leq \delta
\]

if rank($X_0$)$\leq\rho_r n/\mu/log^2 n$ and $k\leq\rho_sn^2$, where $\rho_r$ and $\rho_s$ are positive constant.

\end{theorem}

Similarly to the proof of Theorem~\ref{thm1}, we present the following theorem for the minimization problem of Eq.~\eqref{equ:originForm2}.

\begin{theorem}\label{thm2}

Suppose that $X_0$ satisfies the incoherence condition given by Eq.~\eqref{eq:incoherence}, and $\Omega$ is uniformly distributed among all sets of size $\omega \geq n^2/10$. We assume that each entry is corrupted independently with probability $q$. Let $X$ and $Y$ be the solution of optimization problem given by Eq.~\eqref{equ:originForm2} with parameter $\tau=O(1/\sqrt{n})$ , $\kappa=O(1/\sqrt{n})$ and $\alpha=O(1/n)$. For some constant $c>0$ and sufficiently large $n$, the following holds with probability at least $1-cn^{-10}$,

\[
\|X_0-X\|_F\leq \delta \text{ and }\|Y_0-Y\|_F\leq \delta
\]

if rank($X_0$)$\leq\rho_r n/\mu/log^2 n$ and $q\leq q_s$, where $\rho_r$ and $q_s$ are positive constants.

\end{theorem}

%\subsection{Convergence Analysis}\label{sec:convergence}

We now analyze the global convergence of the proposed UMA algorithm. For simplicity, let $U=(Z;X;Y)$, $V=(X;Y;\Lambda)$, $W=(Z;X;Y;\Lambda)$. We also define
\[
\theta(U) =\|X\|_* + \tau\|Y\|_1-\alpha \langle M, Y \rangle+\frac{\kappa}{2}\|Y\|_F^2
\]
and $\Psi(W)=(-\Lambda; -\Lambda; -\Lambda; X+Y+Z-M)$.
It follows from Corollaries 28.2.2 and 28.3.1 of \cite{rockafellar2015convex} that the solution set of Eq.~(\ref{equ:originForm2}) is non-empty.
Then, let $W^*=((Z^*)^\top,(X^*)^\top,(Y^*)^\top,(\Lambda^*)^\top)^\top$ be a saddle point of Eq.~(\ref{equ:originForm2}), and define $V^*=((X^*)^\top,(Y^*)^\top,(\Lambda^*)^\top)^\top$, ${\cal W}^*$ as the solution set and ${\cal V}^*:=\{V^*| W^*\in {\cal W}^*\}$.

\begin{theorem}\label{theorem34}
	When $\beta$ is restricted by
	\begin{eqnarray}\label{beta}
	\beta\in\big(0,\;2(\sqrt{5}-2)\kappa \big),
	\end{eqnarray} there exists a sufficient small scalar $\varepsilon>0$ such that
	\begin{align*}
	&\label{cof1} \kappa-\frac{\sqrt{5}+2}{2}\beta>0,\;\mbox{and}\;\;\kappa-\frac{1}{2(\sqrt{5}-2-\varepsilon)\beta}>0.
	\end{align*}
	Then the sequence $\{V^k\}$ is bounded and $ \|Y^k - Y^{k+1} \|_F^2 + \|X^k - X^{k+1} \|_F^2 + \| \Lambda^k -{\Lambda}^{k+1}\|_F^2 \to 0$ as $k\rightarrow \infty$. 
	%Then, we have
	%\begin{enumerate}
	%	\item[(1)] The sequence $\{V^k\}$ is bounded.		
	%	\item[(2)] $ \lim_{k\rightarrow \infty}\{\|Y^k - Y^{k+1} \|_F^2 +
	%	\|X^k - X^{k+1} \|_F^2 + \| \Lambda^k -{\Lambda}^{k+1}\|_F^2\}=0.$
	%\end{enumerate}
\end{theorem}

%We are now ready to prove the convergence of UMA.

\begin{theorem}\label{THMD2} Let $\{V^k\}$ and $\{W^{k}\}$
	be the sequences generated by UMA. Assume that  the penalty parameter $\beta$ is satisfied with Eq.~\eqref{beta}.
	Then, we have
	\begin{enumerate}
		
		\item Any cluster point of  $\{{W}^{k}\}$ is a solution point of (\ref{equ:lagrangianForm}).
		
		\item The sequence $\{{V}^k\}$ converges to some $V^{\infty}\in {\cal
			V}^*$.
		\item The sequence $\{U^k\}$ converges to a solution point of Eq.~\eqref{equ:originForm2}.
	\end{enumerate}
\end{theorem}

%\begin{remark}
%	Note that the range for $\beta$ in \cite{tao2016convergence} with convergence guarantee   is
%	$\beta\in(0,0.4\kappa)$ for UMA solving (\ref{equ:originForm2}).
%	However, we get a much larger range for the penalty parameter $\beta$ in (\ref{beta}).
%\end{remark}
Note that the range for $\beta$ in \cite{tao2016convergence} with convergence guarantee   is
$\beta\in(0,0.4\kappa)$ for UMA solving (\ref{equ:originForm2}).
However, we get a much larger range for the penalty parameter $\beta$ in Eq.~(\ref{beta}).

Next, we present a worst-case $O(1/t)$ convergence rate measured by the iteration complexity for UMA.
Indeed, the range of $\beta$ to ensure the $O(1/t)$ convergence rate is slightly more restrictive than Eq.~(\ref{beta}).
Let us define $${U}^{k+1}_t=\frac{1}{t}\sum_{k=1}^t U^{k+1}\;\mbox{and}\; {W}^{k+1}_t=\frac{1}{t}\sum_{k=1}^t W^{k+1}.$$
We define ${X}_t^{k+1}$, ${Y}_t^{k+1}$ and ${Z}_t^{k+1}$ in the same way.
%\begin{eqnarray*}
	%{ Z}^{k+1}_t=\frac{1}{t}\sum_{k=1}^t Z^{k+1},\; { X}_t^{k+1}=\frac{1}{t}\sum_{k=1}^t %X^{k+1},\;{Y}_t^{k+1}=\frac{1}{t}\sum_{k=1}^t Y^{k+1},\;\mbox{and}\; %{U}^{k+1}_t=\frac{1}{t}\sum_{k=1}^t U^{k+1},\;{ W}^{k+1}_t=\frac{1}{t}\sum_{k=1}^t W^{k+1}.
%\end{eqnarray*}
Obviously, ${ W}^{k+1}_t\in{\cal W}$ because of the convexity of ${\cal W}$.
By invoking Theorem \ref{theorem34}, there exists a constant $C$ such that
\begin{eqnarray*}
	\max\left( \|X^k\|_F,\|Y^k\|_F,\|Z^k\|_F,\|\Lambda\|_F\right)\le C, \,\;\forall \;k.
\end{eqnarray*}

\begin{theorem}\label{TA2}
	For $t$ iterations generated by UMA  with $\beta$ restricted in $\beta\in\big(0, (\sqrt{33}-5)\kappa/2\big)$, 
	%\begin{eqnarray}\label{condition:rate}
	%	\beta\in\left(0, \frac{\sqrt{33}-5}{2}\kappa\right),
	%\end{eqnarray} the following assertions holds.
	\begin{itemize}
		\item[1)] We have
		\begin{eqnarray*}
			\lefteqn{\theta(U_t^{k+1})-\theta(U)+(W_t^{k+1}-W)^\top \Psi(W)} \\
			&\le& \frac{1}{t}\big[4\beta C\|X+Y+Z-{P_{\Omega} M}\|_F \\
			&&+\frac{1}{2}\|V^1-V\|_Q^2 +\frac{7+\sqrt{33}}{8}\beta\|Y^1-Y^0\|_F^2\big],
		\end{eqnarray*}
		with $Q=(\beta I,0,0;0,\beta I,0;0,0,I/\beta)$.
		%\begin{eqnarray*}
		%	\theta(U_t^{k+1})-\theta(U)+(W_t^{k+1}-W)^\top \Psi(W) \\
		%	\le \frac{1}{t}\left[4\beta C\|X+Y+Z-{P_{\Omega} M}\|_F
		%	+\frac{1}{2}\|V^1-V\|_Q^2 \right.\\
		%	\left. +\frac{7+\sqrt{33}}{8}\beta\|Y^1-Y^0\|_F^2\right].\quad
		%\end{eqnarray*}
		\item[2)] There exists a constant ${\bar c}_1>0$ such that
		\begin{eqnarray*}
			\|{ X}^{k+1}_t+{ Y}^{k+1}_t+{ Z}^{k+1}_t-{P_{\Omega} M}\|^2\le \frac{{\bar c}_1}{t^2}.
		\end{eqnarray*}
		%\begin{eqnarray}\label{ergoprimal}
		%\end{eqnarray}
		\item[3)] There exists a constant ${\bar c}_2>0$ such that
		\begin{eqnarray*}
			|\theta(U_t^{k+1})-\theta(U^*)|\le\frac{{\bar c}_2}{t}.
		\end{eqnarray*}
	\end{itemize}
\end{theorem}

\begin{table*}[t]
	\caption{Detection precision, recall and F1 compared with other algorithms on \textsf{MovieLens 100K} and \textsf{MovieLens 1M} which are under unorganized malicious attacks based on a combination of traditional strategies.}
	\label{tab:MAMTCompare}
	\vskip 0.15in
	\centering
	%\resizebox{\columnwidth}{!}{
	\begin{tabular}{ccccccc}
		\toprule %\hline\hline
		& \multicolumn{3}{c}{\textsf{Movielens 100K}}   &\multicolumn{3}{c}{\textsf{Movielens 1M}}                 \\
		%\cmidrule{1-1}
		& Precision & Recall & F1 & Precision & Recall & F1  \\
		\midrule
		UMA&\textbf{0.934$\pm$0.003}&\textbf{0.883$\pm$0.019}&\textbf{0.908$\pm$0.011}&\textbf{0.739$\pm$0.009}&\textbf{0.785$\pm$0.023}&\textbf{0.761$\pm$0.016}\\
		RPCA&0.908$\pm$0.010&0.422$\pm$0.048&0.575$\pm$0.047&0.342$\pm$0.003&0.558$\pm$0.028&0.424$\pm$0.009\\
		N-P&0.774$\pm$0.015&0.641$\pm$0.046&0.701$\pm$0.032&0.711$\pm$0.007&0.478$\pm$0.018&0.572$\pm$0.014\\
		k-means&0.723$\pm$0.171&0.224$\pm$0.067&0.341$\pm$0.092&0.000$\pm$0.000&0.000$\pm$0.000&0.000$\pm$0.000\\
		PCAVarSel&0.774$\pm$0.009&0.587$\pm$0.024&0.668$\pm$0.019&0.278$\pm$0.007&0.622$\pm$0.022&0.384$\pm$0.011\\
		MF-based&0.911$\pm$0.009&0.814$\pm$0.008&0.860$\pm$0.009&0.407$\pm$0.005&0.365$\pm$0.004&0.385$\pm$0.005\\
		\bottomrule %\hline\hline
	\end{tabular}
	%}
	\vskip -0.1in
\end{table*}

\begin{table*}[t]
	\caption{Detection precision, recall and F1 compared with other algorithms on \textsf{MovieLens 100K} and \textsf{MovieLens 1M} which are under general unorganized malicious attacks.}
	\label{tab:UMACompare}
	\vskip 0.15in
	\centering
	%\resizebox{\columnwidth}{!}{
	\begin{tabular}{ccccccc}
		\toprule %\hline\hline
		& \multicolumn{3}{c}{\textsf{Movielens 100K}}   &\multicolumn{3}{c}{\textsf{Movielens 1M}}                 \\
		%\cmidrule{1-1}
		& Precision & Recall & F1 & Precision & Recall & F1  \\
		\midrule
		UMA&\textbf{0.929$\pm$0.013}&\textbf{0.865$\pm$0.032}&\textbf{0.896$\pm$0.022}&\textbf{0.857$\pm$0.005}&\textbf{0.733$\pm$0.003}&\textbf{0.790$\pm$0.002}\\
		RPCA&0.797$\pm$0.046&0.659$\pm$0.097&0.721$\pm$0.097&0.635$\pm$0.012&0.391$\pm$0.022&0.484$\pm$0.015\\
		N-P&0.244$\pm$0.124&0.145$\pm$0.089&0.172$\pm$0.084&0.273$\pm$0.020&0.099$\pm$0.031&0.144$\pm$0.035\\
		k-means&0.767$\pm$0.029&0.234$\pm$0.042&0.357$\pm$0.051&0.396$\pm$0.026&0.300$\pm$0.039&0.341$\pm$0.035\\
		PCAVarSel&0.481$\pm$0.027&0.168$\pm$0.017&0.248$\pm$0.023&0.120$\pm$0.006&0.225$\pm$0.012&0.157$\pm$0.008\\
		MF-based&0.556$\pm$0.023&0.496$\pm$0.021&0.524$\pm$0.022&0.294$\pm$0.012&0.264$\pm$0.010&0.278$\pm$0.011\\
		\bottomrule %\hline\hline
	\end{tabular}
	%}
	\vskip -0.1in
\end{table*}

\section{Experiments}\label{experiments}
In this section, we compare our proposed approach UMA with the state-of-the art approaches of attacks detection. 
%In this section, we conduct empirical studies on unorganized malicious attacks detection.

To evaluate the performance of attacks detection, we consider three common evaluating metrics as the performance measures: detection precision, recall and F1~\citep{gunes2014shilling}. Precision measures the fraction of detected attack profiles among all the detected profiles, while recall measures the fraction of detected attack profiles among all the attack profiles. F1 is the harmonic mean of precision and recall. These performance measures can be calculated as follows,
\begin{eqnarray*}
	\text{Precision}&=&\frac{\#\text{true postives}}{\#\text{true postives}+\#\text{false positives}}\\
	\text{Recall}&=&\frac{\#\text{true postives}}{\#\text{true postives}+\#\text{false negatives}}\\
	\text{F1}&=&\frac{2\times \text{precision}\times \text{recall}}{\text{precision}+ \text{recall}}
\end{eqnarray*}
in which $\#$ true positives is the number of attack profiles correctly detected as attacks, $\#$ false positives is the number of normal profiles that are misclassified, and $\#$ false negatives is the number of attack profiles that are misclassified.

\subsection{Datasets}
We first conduct our experiments on the common-used datasets \textsf{MovieLens 100K} and \textsf{MovieLens 1M} collected and released by GroupLens.\footnote[3]{http://grouplens.org/datasets/movielens/.} These datasets are collected from a non-commercial recommender system. It is more likely that the users in this dataset are non-spam users. We take the users already in the datasets as normal users.
%The rating scores are integers from 1 to 5, where 1 and 5 are the worst and best, respectively.  
The rating scores range from 1 to 5, and we preprocess the data by minus 3 to the $[-2,2]$ range.
Dataset \textsf{MovieLens 100K} contains 100000 ratings of 943 users over 1682 movies, and Dataset \textsf{MovieLens 1M} contains 1000209 ratings of 6040 users over 3706 movies. 
We describe how to add attack profiles in Section~\ref{sec:results}.

We also collect a real dataset \textsf{Douban 10K} with 35 attack profiles identified by the Douban website, where registered users record rating information over various films, books, clothes, etc.\footnote[4]{http://www.douban.com/.} We gather 12095 ratings of 213 users over 155 items. 
%where rating scores are integers from 1 to 5. 
The rating scores range from 1 to 5, and we preprocess the data by minus 3 to the $[-2,2]$ range.
Among all the 213 user profiles, 35 profiles are attack profiles.

\subsection{Comparison Methods and Implementation Details}
We compare with the state-of-the-art approaches for attack detection and robust PCA as follows. 
\begin{itemize}
	\item \textbf{N-P:} A statistical algorithm which identifies attack profiles based on the Neyman-Pearson statistics \citep{hurley2009statistical}.
	\item \textbf{k-means:} A cluster algorithm based on several classification attributes \citep{bhaumik2011clustering}. They apply $k$-means clustering based on these attributes and classify users in the smallest cluster as malicious users.
	\item \textbf{PCAVarSel:} A PCA-based variable selection algorithm. They detect malicious users according to the covariance among users \citep{mehta2009unsupervised}.
	\item \textbf{MF-based:} A low-rank matrix factorization based reputation estimation algorithm~\citep{ling2013unified}. They regard low-reputed users as malicious users.
	\item \textbf{RPCA:} A low-rank matrix recovery method considering sparse and small perturbation noise~\citep{candes2011robust}. Users with nonzero terms in the sparse component are classified as malicious users.
\end{itemize}
In all the experiments, we set $\tau={10}/{\sqrt{m}}$, $\alpha=10/m$ and $\delta=\sqrt{{mn}/{200}}$. Because the scale of ratings is from -2 to 2, a rating can be viewed as a malicious rating if it deviates from the ground-truth rating by more than 3; therefore, we set parameter $\beta={\tau}/{3}$ as the entries of $Y$ will be nullified if they are smaller than ${\tau}/{\beta}$ according to Eq.~\eqref{eqnRefreshY}. We set $\kappa=\tau$ which is satisfied with the conditions for convergence, 
i.e., $\beta\in(0, (\sqrt{33}-5)\kappa/2)$.
%i.e., Eq.~\eqref{beta} and Eq.~\eqref{condition:rate}. 
%We use the power method \citep{halko2011finding} to approximate Eq.~\eqref{eqnRefreshX} so that our method has the same scalability as the power method.

\subsection{Comparison Results}\label{sec:results}
In the first experiment, we add attack profiles into the datasets \textsf{MovieLens 100K} and \textsf{MovieLens 1M} by a combination of several traditional attack strategies. These traditional attack strategies include average attack strategy, random attack strategy and bandwagon attack strategy which have been discussed in Section~\ref{problem-form}. Specifically, each attacker randomly chooses one strategy to fake the user rating profiles and promotes one item randomly selected from items with average rating lower than 0.
%As mentioned before, the datasets \textsf{MovieLens 100K} and \textsf{MovieLens 1M} do not contain attackers; therefore, we first use a combination of several traditional attack strategies to construct the datasets with unorganized malicious attacks. These traditional attack strategies include average attack strategy, random attack strategy and bandwagon attack strategy~\citep{lam2004shilling}. Here, each attacker randomly chooses one strategy to fake the user rating profiles and promote one item whose average rating is lower than 2. 
In line with the setting of previous attack detection works, we set the filler ratio (percentage of rated items in total items) as $0.01$ and the filler items are drawn from the top 10\% most popular items. We set the spam ratio (number of attack profiles$/$number of all user profiles) as $0.2$. Table~\ref{tab:MAMTCompare} shows the experimental results on datasets \textsf{MovieLens 100K} and \textsf{MovieLens 1M} which are under unorganized malicious attacks based on a combination of traditional strategies.

\begin{table}[t]
	\caption{Detection precision, recall and F1 compared with other algorithms on dataset
		\textsf{Douban 10K}.}
	\label{tab:CompareDouban}
	\vskip 0.15in
	\centering
	%	\resizebox{\columnwidth}{!}{
	\begin{tabular}{*{4}{c}}
		\toprule
		Methods & Precision& Recall& F1  \\
		\midrule
		UMA &\textbf{0.800}&\textbf{0.914}&\textbf{0.853}\\
		RPCA &0.535&0.472&0.502\\
		N-P &0.250&0.200&0.222\\
		k-means &0.321&0.514&0.396\\
		PCAVarSel &0.240&0.343&0.282\\
		MF-based &0.767&0.657&0.708\\
		\bottomrule
	\end{tabular}
	%	}
	\vskip -0.15in
\end{table}

In the second experiment, we consider a more general case of unorganized malicious attacks. Besides the profile injection attacks as mentioned above, attackers can hire existing users to attack their targets. So we add these two kinds of attacks into the datasets. We set spam ratio as 0.2, where 25\% of the attack profiles are produced similar to the first experiment, and 75\% of the attack profiles are from existing users by randomly changing the rating of one item lower than 0 to 2. In this case, attacks are more difficult to be detected, because the attack profiles are more similar to normal user profiles.
%Besides inserting new user profiles into the rating matrix, attackers can hire existing users to attack. We reconstruct datasets from \textsf{MovieLens 100K} and \textsf{MovieLens 1M} as follows. We set the filler ratio as $0.01$, and the filler items are drawn from the top 10\% most popular items. Most attack profiles (75\%) are generated by i) randomly selecting an item with average rating lower than 2, ii) randomly selecting a user profile with a rating lower than 2 to the chosen item, and iii) the selected user profile can be viewed as an attack profile by modifying the rating to 5. The other attack profiles (25\%) are produced similar to that of the above paragraph. 
Table~\ref{tab:UMACompare} demonstrates the comparison results on datasets \textsf{MovieLens 100K} and \textsf{MovieLens 1M} which are under general unorganized malicious attacks.

Table~\ref{tab:CompareDouban} shows the experiments on dataset \textsf{Douban 10K}. The experimental results in Table~\ref{tab:MAMTCompare}, \ref{tab:UMACompare} and \ref{tab:CompareDouban} show that our proposed UMA approach achieves the best performance on three measures: Precision, Recall and F1, and our approach takes superiority in all datasets. It is easy to observe that traditional attack detection approaches fail to work for unorganized malicious attacks detection, because those methods depend on the properties of shilling attacks, e.g., the k-means method and N-P method work if the attack profiles are similar in the view of the supervised features or their latent categories, and the PCAVarSel method achieves good performance only if attack profiles have more common unrated items than normal profiles. In summary, these methods detect attacks by identifying some common characteristics of the attack profiles. However, those good properties do not exist in the unorganized malicious attacks.

The RPCA and MF-based methods try to find the ground-truth rating matrix from the observed rating matrix, whereas they hardly separate the sparse attack matrix from the noisy matrix, and tend to suffer from low precision, especially on large-scale and heavily sparse datasets.

Since different systems may contain different spam ratios (number of attack profiles$/$number of all user profiles), we compare our UMA algorithm with other methods by varying the spam ratio from 2\% to 20\% in Figure~\ref{SpamRatio}. Our UMA approach achieves robust and better performance in different spam ratios, whereas the comparison methods (except the RPCA method) achieve worse performance for small spam ratio, e.g., the N-P approach detects almost nothing. Although the RPCA method is as stable as our method UMA in different spam ratios, there is a performance gap between RPCA and UMA which becomes bigger when the dataset gets larger and sparser from \textsf{MovieLens 100K} to \textsf{MovieLens 1M}.

\begin{figure}[t]
		\vskip 0.15in
	\begin{center}
		\centerline{\includegraphics[width=1.1\columnwidth]{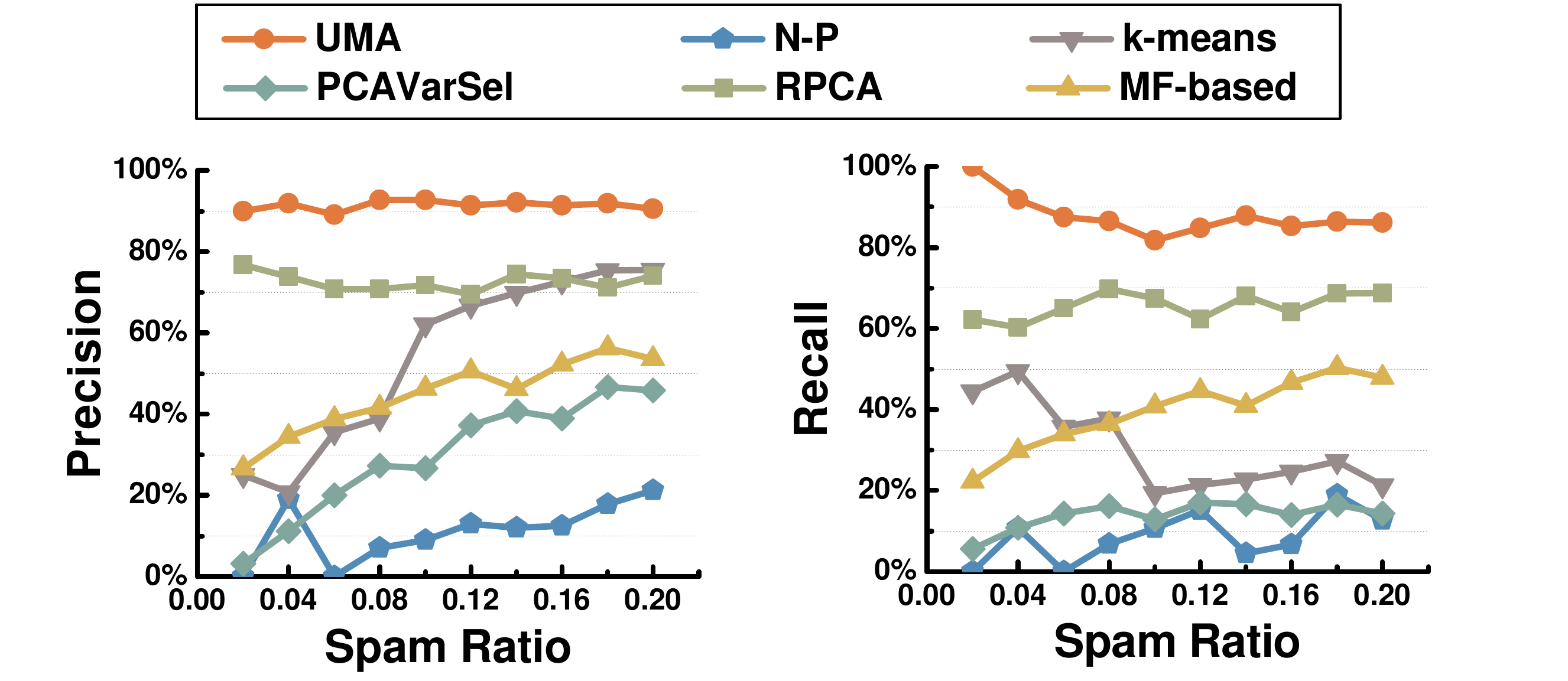}}
		%\vspace{-0.15in}
		%\vskip -0.2in
		\caption{Detection precision and recall on \textsf{MovieLens 100K} under unorganized malicious attacks. The spam ratio (number of attack profiles$/$number of all user profiles) varies from 0.02 to 0.2.}
		\label{SpamRatio}
	\end{center}
	\vskip -0.35in
\end{figure}

\section{Conclusion}
The attack detection plays an important role to improve the quality of recommendation, whereas most previous methods focus on shilling attacks. 
The key to detect shilling attacks is to find the common characteristics of the attack profiles which are produced by the same attack strategy. However, unorganized malicious attacks are produced by multiple attack strategies to attack different targets. In this case, previous works perform inefficiently.
In this paper, we first formulate the unorganized malicious attacks detection as a variant of matrix completion problem. Then we propose the UMA algorithm, and give the proof of its recovery guarantee and global convergence. Experimental results show that UMA achieves significantly better performance than the state-of-the-art approaches of attack detection.

%The attack detection plays an important role to improve the quality of recommendation in systems, whereas most previous methods focus on shilling attacks. In this paper, we first formulate the unorganized malicious attacks detection as a variant of matrix completion problem. Then we propose the Unorganized Malicious Attacks detection (UMA) algorithm, which can be viewed as a proximal alternating splitting augmented Lagrangian method. We give the proof of global convergence theoretically, and experimental results show that our proposed algorithm achieves significantly better performance than the state-of-the-art approaches for attack detection.

\nocite{langley00}

\bibliography{nips_2016_2}
\bibliographystyle{icml2018}

%%%%%%%%%%%%%%%%%%%%%%%%%%%%%%%%%%%%%%%%%%%%%%%%%%%%%%%%%%%%%%%%%%%%%%%%%%%%%%%
%%%%%%%%%%%%%%%%%%%%%%%%%%%%%%%%%%%%%%%%%%%%%%%%%%%%%%%%%%%%%%%%%%%%%%%%%%%%%%%

\newpage

{\onecolumn
\appendix
\begin{center}
{\Large\textbf{Supplementary Material (Appendix)}}
\end{center}

\section{Optimization Model of UMA}

We consider the following optimization model:
\begin{eqnarray}
	\label{CouplePCA}
	\begin{array}{cl}
		\min&\|X\|_*+\tau\|Y\|_1-\alpha\langle M,Y\rangle+\frac{\kappa}{2}\|Y\|_F^2,\\
		s.t.&X+Y+Z={\bar M},\\
		&Z\in{\bf B}, \\
		& {\bf B}:=\{Z|\|P_{\Omega}(Z)\|_F\le \delta\},
	\end{array}
	\begin{array}{cc}
	\end{array}
\end{eqnarray}
where    $\kappa>0$ is a regularization parameter and ${\bar M}:=P_{\Omega}(M)$.
The model (\ref{CouplePCA}) is a three-block convex programming. 
We define the Lagrangian function  and 
augmented Lagrangian function of (\ref{CouplePCA}) as follows:
\begin{equation}\label{AL-Fun}
{\cal L}(X,Y,Z,\Lambda,\beta):=\|X\|_{*}+\tau \|Y\|_1-\alpha\langle M,Y\rangle+\frac{\kappa}{2}\|Y\|_F^2
-\langle \Lambda,X+Y+Z-{\bar M} \rangle,
\end{equation}
\begin{equation}\label{ALF}
{\cal L}_{\mathcal{A}}(X,Y,Z,\Lambda,\beta):=\|X\|_{*}+\tau \|Y\|_1-\alpha\langle M,Y\rangle+\frac{\kappa}{2}\|Y\|_F^2
-\langle \Lambda,X+Y+Z-{\bar M} \rangle +{ \beta \over 2} \|X+Y+Z-{\bar M}\|_F^2,
\end{equation}
where $\beta>0$ is the penalty parameter.

\section{Recovery Guarantee}
In this section, we present theoretical guarantee that UMA can recover the low-rank component $X$ and the sparse component $Y$. 
For simplicity, our theoretical analysis focuses on square matrix, and it is natural to generalize our results to the general rectangular matrices. 

Let the singular value decomposition of $X_0\in{\cal R}^{n\times n}$ be given by
\[
X_0=S\Sigma D^\top=\sum\nolimits_{i=1}^r \sigma_is_id_i^\top
\]
where $r$ is the rank of matrix $X_0$, $\sigma_1,\ldots,\sigma_r$ are the positive singular values, and $S=[s_1,\ldots,s_r]$ and $D=[d_1,\ldots,d_r]$ are the left- and right-singular matrices, respectively. For $\mu>0$, we assume
\begin{eqnarray}\label{eq:incoherence}
	\max_i\|S^\top e_i\|^2&\leq& {\mu r}/{n},\nonumber\\
	\max_i\|D^\top e_i\|^2&\leq& {\mu r}/{n},\\
	\|SD^\top\|_\infty^2&\leq& \mu r/n^2.\nonumber
\end{eqnarray}

Firstly, we consider the following optimization problem where all the entries of $M$ can be observed.
\begin{equation}\label{equ:originForm}
	\begin{split}
		\min_{X,Y,Z}\,&\|X\|_*+\tau\|Y\|_1-\alpha \langle M,Y\rangle+\frac{\kappa}{2}\|Y\|^2_F\\
		\text{s.t. }&X+Y+Z=M, \\
		&\|Z\|_F\leq \delta.
	\end{split}
\end{equation}

\begin{theorem}\label{thm1}
	Suppose that the support set of $Y_0$ be uniformly distributed for all sets of cardinality $k$,  and $X_0$ satisfies the incoherence condition given by Eqn.~\eqref{eq:incoherence}. Let $X$ and $Y$ be the solution of optimization problem given by Eqn.~\eqref{equ:originForm} with parameter $\tau=O(1/\sqrt{n})$ , $\kappa=O(1/\sqrt{n})$ and $\alpha=O(1/n)$. For some constant $c>0$ and sufficiently large $n$, the following holds with probability at least $1-cn^{-10}$ over the choice on the support of $Y_0$
	\[
	\|X_0-X\|\leq \delta \text{ and }\|Y_0-Y\|_F\leq \delta
	\]
	if rank($X_0$)$\leq\rho_r n/\mu/log^2 n$ and $k\leq\rho_sn^2$, where $\rho_r$ and $\rho_s$ are positive constant.
\end{theorem}

\begin{proof}
	
	Let $\Omega$ be the space of matrices with the same support as $Y_0$, and let $T$ denote the linear space of matrices
	
	\[
	T:=\{SX^*+YD^*, X,Y\in\mathbb{R}^{n\times r}\}.
	\]
	
	We will first prove that, for $\|P_\Omega P_T\|\leq1/2$, $(X_0,Y_0)$ is the unique solution if there is a pair $(W,F)$ satisfying
	
	\begin{equation}\label{eq:assumption}
		SD^*+W=\tau(\text{sgn}(Y_0)+F+P_\Omega K)
	\end{equation}
	
	where $P_TW=0$ and $\|W\|\leq 1/2$, $P_\Omega F=0$ and $\|F\|_\infty\leq 1/2$ and $\|P_\Omega K\|_F\leq 1/4$. Notice that $SD^*+W_0$ is an arbitrary subgradient of $\|X\|_*$ at $(X_0,Y_0)$, and $\tau(\text{sgn}(Y_0)+F_0)-\alpha M +\kappa Y_0$ is an arbitrary subgradient of $\tau\|Y\|_1-\alpha \langle M,Y\rangle +\kappa\|Y\|_F^2/2$ at $(X_0,Y_0)$. For any matrix $H$, we have, by the definition of subgradient,
	
	\begin{multline}\label{eq:pf:tmp1}
		\|X_0+H\|_*+\tau\|Y_0-H\|_1-\alpha \langle M,Y_0-H\rangle +\tfrac{\kappa}{2}\|Y_0-H\|_F^2\\
		\geq\|X_0\|_*+\tau\|Y_0\|_1-\alpha \langle M,Y_0\rangle+\tfrac{\kappa}{2}\|Y_0\|_F^2+\langle \alpha M-\kappa Y_0,H\rangle\\
		+\langle SD^*+W_0,H\rangle- \tau\langle \text{sgn}(Y_0)+ F_0, H\rangle.
	\end{multline}
	
	By setting $W_0$ and $F_0$ satisfying $\langle W_0,H\rangle=\|P_{T^\bot}H\|_*$ and $\langle F_0,H\rangle=-\|P_{\Omega^\bot}H\|_1$, we have
	
	\begin{eqnarray}
		\lefteqn{\langle SD^*+W_0,H\rangle- \tau\langle \text{sgn}(Y_0)+ F_0, H\rangle}\nonumber\\
		&=&\|P_{T^\bot}H\|_*+\tau \|P_{\Omega^\bot}H\|_1+\langle SD^*-\tau \text{sgn}(Y_0),H\rangle\nonumber\\
		&=&\|P_{T^\bot}H\|_*+\tau \|P_{\Omega^\bot}H\|_1+\langle \tau(F+P_\Omega K)-W,H\rangle\nonumber\\
		&\geq&\frac{1}{2}(\|P_{T^\bot}H\|_*+\tau \|P_{\Omega^\bot}H\|_1)+\tau\langle P_\Omega D,H\rangle\label{eq:pf:tmp2}
	\end{eqnarray}
	
	where the second equality holds from Eqn.~\eqref{eq:assumption}, and the last inequality holds from
	
	\begin{equation*}
		\langle \tau F-W,H\rangle\geq -|\langle W,H\rangle|-|\langle \tau F,H\rangle|
		\geq -(\|P_{T^\bot}H\|_*+\tau \|P_{\Omega^\bot}H\|_1)/2
	\end{equation*}
	
	for $\|W\|\leq 1/2$ and $\|F\|_\infty\leq 1/2$. We further have
	
	\begin{equation}\label{eq:pf:tmp3}
		\langle \tau P_\Omega K,H\rangle\geq  -\frac{\tau}{4}\|P_{\Omega^{\bot}}H\|_F-\frac{\tau}{2}\|P_{T^{\bot}}H\|_F
	\end{equation}
	
	from $\|P_\Omega K\|_F\leq 1/4$ and
	
	\begin{multline*}
		\|P_\Omega H\|_F \leq \|P_\Omega P_T H\|_F+\|P_\Omega P_{T^\bot} H\|_F\leq\|P_\Omega P_{T^\bot} H\|_F\\
		+\|H\|_F/2 \leq (\|P_\Omega H\|_F+\|P_{\Omega^\bot} H\|_F)/2+ \|P_\Omega P_{T^\bot} H\|_F.
	\end{multline*}
	
	Combing with Eqns.~\eqref{eq:pf:tmp1} to \eqref{eq:pf:tmp3}, we have
	
	\begin{eqnarray*}
		\lefteqn{\|X_0+H\|_*+\tau\|Y_0-H\|_1-\alpha \langle M,Y_0-H\rangle+\tfrac{\kappa}{2}\|Y_0-H\|_F^2}\\
		&\geq&\|X_0\|_*+\tau\|Y_0\|_1-\alpha \langle M,Y_0\rangle+\tfrac{\kappa}{2}\|Y_0-H\|_F^2\\
		&&+\langle \alpha M-\kappa Y_0,H\rangle+\tfrac{1-\tau}{2}\|P_{T^{\bot}}H\|_*+\tfrac{\tau}{4}\|P_{\Omega^{\bot}}H\|_1
	\end{eqnarray*}
	
	From the conditions that $\Omega \cap T=\{0\}$, $\tau=O(1/\sqrt{n})$, $\kappa=O(1/\sqrt{n})$ and $\alpha=O(1/n)$, we have
	
	\[
	\alpha\langle Y_0-X_0,H\rangle +\frac{1-\tau}{2}\|P_{T^{\bot}}H\|_*+\frac{\tau}{4}\|P_{\Omega^{\bot}}H\|_1>0
	\]
	
	for sufficient large $n$. Therefore, we can recover $X_0$ and $Y_0$ if there is a pair $(W,F)$ satisfying Eqn.~\eqref{eq:assumption}, and the pair $(W,F)$ can be easily constructed according to \cite{candes2011robust}. We complete the proof from the condition $\|Z\|_F\leq \delta$.
	
\end{proof}

Similarly to the proof of Theorem~\ref{thm1}, we present the following theorem for the minimization problem of Eqn.~\eqref{CouplePCA}.
\begin{theorem}\label{thm2}
	Suppose that $X_0$ satisfies the incoherence condition given by Eqn.~\eqref{eq:incoherence}, and $\Omega$ is uniformly distributed among all sets of size $m\geq n^2/10$. We assume that each entry is corrupted independently with probability $q$. Let $X$ and $Y$ be the solution of optimization problem given by Eqn.~\eqref{CouplePCA} with parameter $\tau=O(1/\sqrt{n})$ , $\kappa=O(1/\sqrt{n})$ and $\alpha=O(1/n)$. For some constant $c>0$ and sufficiently large $n$, the following holds with probability at least $1-cn^{-10}$
	\[
	\|X_0-X\|_F\leq \delta \text{ and }\|Y_0-Y\|_F\leq \delta
	\]
	if rank($X_0$)$\leq\rho_r n/\mu/log^2 n$ and $q\leq q_s$, where $\rho_r$ and $q_s$ are positive constants.
\end{theorem}

\section{Optimality condition}

Before starting to show the convergence, we  derive its optimality condition of (\ref{CouplePCA}).
Let ${\cal W}:={\bf B} \times {\cal R}^{m\times n} \times {\cal R}^{m\times n}
\times {\cal R}^{m\times n}$. It follows from Corollaries 28.2.2 and 28.3.1 of \cite{rockafellar2015convex} that the solution set of (\ref{CouplePCA}) is non-empty.
Then, let $W^*=((Z^*)^\top,(X^*)^\top,(Y^*)^\top,(\Lambda^*)^\top)^\top$ be a saddle point of (\ref{CouplePCA}).
It is easy to see that (\ref{CouplePCA}) is
equivalent to finding $W^* \in {\cal
	W}$ such that
\begin{equation}
\label{MA-VIxyl}
\begin{array}{ll}
\left\{
\begin{array}{l}
\langle Z- Z^*,  - \Lambda^*\rangle \ge 0, \\
\|X\|_*-\|X^*\|_* + \langle X- X^*, -\Lambda^*\rangle \ge 0, \\
\tau\|Y\|_1-\tau\|Y^*\|_1 + \langle Y- Y^*, -\alpha M +\kappa Y^*-\Lambda^*\rangle \ge0, \\
X^*+Y^*+Z^*-{\bar M}=0,
\end{array} \right.
& \;\; \forall\; W=(Z^\top,X^\top,Y^\top,\Lambda^\top)^\top\in  {\cal W},
\end{array}
\end{equation}
or, in a more compact form:
\begin{subequations}\label{MSVII}
	\[\label{MSVI-F}
	\hbox{VI}({\cal W},\Psi,\theta) \qquad \theta(U) - \theta(U^*)+  \langle W-W^*, \Psi(W^*) \rangle \ge \frac{\kappa}{2}\|Y-Y^*\|_F^2, \quad \forall \
	W\in  {\cal W}, \qquad \qquad
	\]
	where
	\[ \label{MD-F-1}
	U=\left(  \begin{array}{c} Z \\ X
	\\Y
	\end{array}
	\right),
	\quad \theta(U) =\|X\|_* + \tau\|Y\|_1-\alpha \langle M, Y \rangle+\frac{\kappa}{2}\|Y\|_F^2,
	\]
	\[\label{MD-F-2}
	\mbox{and} \;\;  W=\left(  \begin{array}{c} Z \\ X
	\\Y\\
	\Lambda \end{array}
	\right), \quad
	V=\left(  \begin{array}{c}  X
	\\Y\\
	\Lambda \end{array}
	\right),  \quad \quad
	\Psi(W)=\left(  \begin{array}{c}-\Lambda \\
	-\Lambda\\
	-\Lambda\\
	X+Y+Z-{\bar M} \end{array}
	\right).
	\]
\end{subequations}
Note that $U$ collects all the primal variables in
(\ref{MA-VIxyl}) and it is a sub-vector of $W$. Moreover, we use $\cal W^*$   to denote the solution set of
$\hbox{VI}({\cal W},\Psi,\theta)$ and define $V^*=((X^*)^\top,(Y^*)^\top,(\Lambda^*)^\top)^\top$ and ${\cal V}^*:=\{V^*| W^*\in {\cal W}^*\}$.

\section{Convergence Analysis}

\setcounter{equation}{0}

In this section, we solve (\ref{CouplePCA}) with global convergence.
More specifically, let $(X^{k},Y^{k},\Lambda^k)$ be given, UMA
generates the new iterate $W^{k+1}$
via the following scheme:
\begin{eqnarray}\label{GADM}
	\left\{
	\begin{array}{ll}
		Z^{k+1} = \arg\min_{Z \in {\bf B}}
		\mathcal{L}_{\mathcal{A}}(X^k,Y^{k},Z,\Lambda^k,\beta),\\
		X^{k+1} = \arg\min_{X\in {\cal R}^{m \times n}}  \mathcal{L}_{\mathcal{A}}(X,Y^k,Z^{k+1},\Lambda^{k},\beta),\\
		Y^{k+1} = \arg\min_{Y \in {\cal R}^{m \times n}}  \mathcal{L}_{\mathcal{A}}(X^{k+1},Y,Z^{k+1},\Lambda^{k},\beta),\\
		\Lambda^{k+1} = \Lambda^k - \beta(X^{k+1} + Y^{k+1} +Z^{k+1}
		-{\bar M}),
	\end{array}
	\right.
\end{eqnarray}
which can be easily written into the following more specific form:
\begin{subnumcases}{\label{ADM2}}
	Z^{k+1}= {\arg\min}_{Z \in {\bf B}}\frac{\beta}{2}\|Z+X^k+Y^{k}-{1 \over \beta}\Lambda^k-{\bar M}\|^2_F, \label{Zsub}\\
	X^{k+1}=  {\arg\min}_{X \in {\cal R}^{m \times n}}\|X\|_*+\frac{\beta}{2}\|X+Y^k+Z^{k+1}-{1 \over \beta}\Lambda^k-{\bar M}\|_F^2, \label{Xsub}\\
	Y^{k+1}=  {\arg\min}_{Y \in  {\cal R}^{m \times n}}\tau\|Y\|_{1}-\alpha \langle M,Y \rangle+\frac{\kappa}{2}\|Y\|_F^2+\frac{\beta}{2}\|Y+X^{k+1}+Z^{k+1}-\frac{1}{\beta}\Lambda^k-{\bar M}\|_F^2,\label{Ysub}\\
	\Lambda^{k+1}= \Lambda^{k}-\beta(X^{k+1}+Y^{k+1}+Z^{k+1}-{\bar M}\overline{}). \label{Lsub}
\end{subnumcases}

In the following, we concentrate on the convergence of UMA. In contrast to the existing results in \cite{tao2016convergence}, we aim to present a much more sharp result.
We first prove some  properties of the
sequence generated by UMA, which play a crucial role
in the coming convergence analysis.
Before that, we introduce some notations:
\begin{align}
	&\label{Dell} \Delta_{\Lambda}:=\frac{1}{2\beta}(\|\Lambda^{k+1}-\Lambda^*\|_F^2-\|\Lambda^k-\Lambda^*\|_F^2+\|\Lambda^{k+1}-\Lambda^k\|_F^2),\\
	&\label{Delx} \Delta_{X}:=\frac{1}{2\beta}(\|X^{k+1}-X^*\|_F^2-\|X^k-X^*\|_F^2+\|X^{k+1}-X^k\|_F^2),\\
	&\label{Dely} \Delta_{Y}:=\frac{1}{2\beta}(\|Y^{k+1}-Y^*\|_F^2-\|Y^k-Y^*\|_F^2+\|Y^{k+1}-Y^k\|_F^2),\\
	&\label{Rg}   {\cal R}=X+Y+Z-{\bar M},\\
	&\label{Rk} {\cal R}^{k+1}=X^{k+1}+Y^{k+1}+Z^{k+1}-{\bar M}.
\end{align}

\begin{lemma}\label{Lem30}
	Let $\{W^k\}$ be generated by UMA.
	Then, we have
	\begin{itemize}
		\item[(1)]
		\begin{eqnarray}
			\label{LY}
			\langle \Lambda^k-\Lambda^{k+1}, Y^k-Y^{k+1}\rangle \ge\kappa\|Y^k-Y^{k+1}\|_F^2.
		\end{eqnarray}
		\item[(2)]
		\begin{eqnarray}
			\label{LX}
			\langle \Lambda^k-\Lambda^{k+1}, X^k-X^{k+1} \rangle\ge -\beta\langle X^k-X^{k+1}, Y^{k+1}-Y^k-(Y^k-Y^{k-1}) \rangle
		\end{eqnarray}
	\end{itemize}
\end{lemma}

\begin{proof}
	(1) Using the optimality of (\ref{Ysub}), we get
	\begin{eqnarray}\label{Yopt}
		\langle Y-Y^{k+1}, \partial (\tau\|Y^{k+1}\|_1)-\Lambda^{k+1}-\alpha M+\kappa Y^{k+1}\rangle \ge 0.
	\end{eqnarray}
	Setting $Y:=Y^k$ in (\ref{Yopt}), we have
	\begin{eqnarray}\label{kkp1}
		\langle Y^k-Y^{k+1}, \partial (\tau\|Y^{k+1}\|_1)-\Lambda^{k+1}-\alpha M+\kappa Y^{k+1}\rangle \ge 0.
	\end{eqnarray}
	Then, setting $Y:=Y^{k+1}$ in (\ref{Yopt}) with the index $k$ replaced with $k-1$, it yields
	\begin{eqnarray}\label{kp1k}
		\langle Y^{k+1}-Y^{k}, \partial (\tau\|Y^{k}\|_1)-\Lambda^{k}-\alpha M+\kappa Y^{k}\rangle \ge 0.
	\end{eqnarray}
	Thus, adding (\ref{kkp1}) and (\ref{kp1k}) together, the inequality (\ref{LY}) follows directly.\\
	
	(2) The inequality (\ref{LX}) can be proved in a similar way as (\ref{LY}).
\end{proof}
\medskip
\begin{lemma}\label{Lem31}  Let $\{W^k\}$ be generated by UMA.
	Then, we have the
	following inequality:
	\begin{eqnarray}
		\label{solstr}
		\lefteqn{\theta(U) - \theta(U^{k+1})+  \langle W-W^{k+1}, \Psi(W^{k+1})\rangle+\beta\langle {\cal R}, \Gamma(X^k,Y^k,Z^k) \rangle}\nonumber\\
		&& \ge
		\frac{\beta}{2}(\|V^{k+1}-V\|_Q^2+\|V^k-V^{k+1}\|^2_Q-\|V^k-V\|_Q^2)+\kappa\|Y^{k+1}-Y^k\|_F^2+\frac{\kappa}{2}\|Y^{k+1}-Y\|_F^2\nonumber\\
		&&-\beta\langle X^{k+1}-X^k,Y^{k+1}-Y^k-(Y^k-Y^{k-1}) \rangle\nonumber\\
		&&+\beta \langle Y^{k+1}-Y, X^{k+1}-X^k \rangle.
	\end{eqnarray}
	where
	\begin{align}
		&\Gamma(X^k,Y^k,Z^k)=Y^k-Y^{k+1}+X^k-X^{k+1},\nonumber\\
		&Q=\left(\begin{array}{ccc}
			\beta I&0&0\\
			0&\beta I&0\\
			0&0&\frac{1}{\beta}I
		\end{array}\right)
	\end{align}

\end{lemma}
\begin{proof}
	According to the optimality condition of  (\ref{GADM}), we have
	\begin{eqnarray} \label{optwhol}
		\begin{array}{ll}
			\left\{
			\begin{array}{l}
				\langle Z- Z^{k+1},  - \Lambda^{k+1}+\beta(X^k-X^{k+1})+\beta(Y^k-Y^{k+1})\rangle \ge 0, \\
				\|X\|_*-\|X^{k+1}\|_* + \langle X- X^{k+1}, -\Lambda^{k+1}+\beta(Y^k-Y^{k+1}\rangle \ge 0, \\
				\tau\|Y\|_1-\tau\|Y^{k+1}\|_1 + \langle Y- Y^{k+1}, -\alpha M +\kappa Y^{k+1}-\Lambda^{k+1}\rangle \ge 0, \\
				\langle \Lambda -\Lambda^{k+1}, X^{k+1}+Y^{k+1}+Z^{k+1}-{\bar M}-\frac{1}{\beta}(\Lambda^k-\Lambda^{k+1})\rangle\ge0,
			\end{array} \right.
			& \forall\; W=(Z^\top,X^\top,Y^\top,\Lambda^\top)^\top\in  {\cal W}.
		\end{array}
	\end{eqnarray}
	Then, combining the above inequalities with (\ref{MD-F-1}) and (\ref{MD-F-2}), we get
	\begin{eqnarray}\nonumber
		\lefteqn{\theta(U) - \theta(U^{k+1})+  \langle W-W^{k+1}, \Psi(W^{k+1})\rangle+\beta\left(\langle Z-Z^{k+1}, Y^k-Y^{k+1}+X^k-X^{k+1} \rangle\right.}\nonumber\\
		&&\left.+\beta\langle X-X^{k+1}, Y^k-Y^{k+1}\rangle\right)\ge\frac{1}{2\beta}\Delta_{\Lambda}+\frac{\kappa}{2}\|Y-Y^{k+1}\|_F^2.\nonumber
	\end{eqnarray}
	Then, invoking (\ref{Rg}) and (\ref{Rk}), we obtain that
	\begin{eqnarray*}
		\lefteqn{ \theta(U) - \theta(U^{k+1})+  \langle W-W^{k+1}, \Psi(W^{k+1})\rangle+\beta\langle {\cal R}-{\cal R}^{k+1}, Y^k-Y^{k+1}+X^k-X^{k+1}\rangle}\nonumber\\
		&&\ge\frac{\kappa}{2}\|Y-Y^{k+1}\|_F^2+\frac{1}{2\beta}\Delta_{\Lambda}+\frac{\beta}{2}(\Delta_{X}+\Delta_{Y})+\beta\langle Y-Y^{k+1},X^k-X^{k+1} \rangle.
	\end{eqnarray*}
	Thus, using  ${\cal R}^{k+1}=\frac{1}{\beta}(\Lambda^k-\Lambda^{k+1})$, it yields that 
	\begin{eqnarray}\label{Lem32pr1}
		\lefteqn{ \theta(U) - \theta(U^{k+1})+  \langle W-W^{k+1}, \Psi(W^{k+1})\rangle+\beta\langle {\cal R}, Y^k-Y^{k+1}+X^k-X^{k+1}\rangle}\nonumber\\
		&\ge&\frac{\kappa}{2}\|Y-Y^{k+1}\|_F^2+\frac{1}{2\beta}\Delta_{\Lambda}+\frac{\beta}{2}(\Delta_{X}+\Delta_{Y})+\beta\langle Y-Y^{k+1},X^k-X^{k+1} \rangle\nonumber\\
		&&+\frac{1}{\beta}\langle \Lambda^k-\Lambda^{k+1},Y^k-Y^{k+1}+X^k-X^{k+1}\rangle.
	\end{eqnarray}
	On the other hand, adding (\ref{LY}) and (\ref{LX}) together, we obtain that
	\begin{eqnarray*}
		\langle \Lambda^k-\Lambda^{k+1}, Y^k-Y^{k+1}+X^k-X^{k+1}\rangle \ge\kappa\|Y^k-Y^{k+1}\|_F^2-\beta\langle X^{k+1}-X^k, Y^{k+1}-Y^k-(Y^k-Y^{k-1}) \rangle.
	\end{eqnarray*}
	Next, substituting the above inequality into (\ref{Lem32pr1}),  and  invoking,
	it yields the assertion (\ref{solstr}).

	\qed
\end{proof}

In the following, we  give each crossing term in the right-hand of (\ref{solstr}) a low bound.
The following inequalities enable us to get a much  sharper result for UMA solving (\ref{CouplePCA}) in contrast to (\cite{tao2016convergence}).

\begin{lemma}\label{pote}
	Let $\{W^k\}$ be generated by UMA.
	Suppose that $0<\varepsilon <\sqrt{5}-2$.
	Then, it holds that
	\begin{align}
		&-\beta\langle X^{k+1}- X^{k},Y^{k+1}-Y^{k} \rangle \ge \beta\left(-\frac{3-\sqrt{5}}{4}\|X^k-X^{k+1}\|_F^2 -\frac{1}{3-\sqrt{5}}\|Y^{k+1}-Y^k\|_F^2\right),\label{Cau1}\\
		&\beta\langle X^{k+1}-X^k, (Y^k-Y^{k-1})\rangle\ge \beta\left(-\frac{3-\sqrt{5}}{4}\|X^k-X^{k+1}\|_F^2  -\frac{1}{3-\sqrt{5}} \|Y^k-Y^{k-1}\|_F^2 \right),\label{Cau2}\\
		&\beta \langle Y^{k+1}-Y, X^{k+1}-X^k \rangle\ge- \beta \left( \frac{1}{2(\sqrt{5}-2-\varepsilon)}\|Y^{k+1}-Y\|_F^2+\frac{\sqrt{5}-2-\varepsilon}{2}\|X^{k+1}-X^k\|_F^2\right).\label{Cau3}
	\end{align}
\end{lemma}
\begin{proof}
	These three inequalities follow from Cauchy-Schwarz inequality.\qed
	%By setting $Y=Y^k$ in the third inequality of (\ref{Lem-1}) we get
	%$$ \tau\|Y^k\|_1-\tau\|Y^{k+1}\|_1+\langle Y^k- { Y}^{k+1}, \alpha X^{k+1}+\kappa Y^{k+1}-\Lambda^{k+1}\rangle\ge0.
	%$$
	%Similarly, taking $k :=k-1$ and $Y=Y^{k+1}$  in the third inequality of (\ref{Lem-1}) we have
	%$$ \tau\|Y^{k+1}\|_1-\tau\|Y^{k}\|_1+\langle Y^{k+1}- { Y}^{k}, \alpha X^{k}+\kappa Y^{k}-\Lambda^{k}\rangle\ge0.
	%$$
	%By adding the above two inequalities, we obtain the conclusion.
\end{proof}

\begin{theorem}\label{theorem33} Let $\{W^k\}$ be generated by UMA.
	Assume that   $\beta>0$ in Algorithm (\ref{GADM}).
	Suppose that $0<\varepsilon <\sqrt{5}-2$.
	%$$\beta+\kappa-\frac{\alpha}{4(\kappa-\alpha)}>0.$$
	Then, we have the
	following contractive property:
	\begin{eqnarray}
		\label{Contr}
		\lefteqn{\frac{\beta}{2}\|X^{k+1}-X^*\|_F^2+\frac{\beta}{2}\|Y^{k+1}-Y^*\|_F^2+\frac{1}{2\beta}\|Z^{k+1}-Z^*\|_F^2+\frac{\beta}{3-\sqrt{5}}\|Y^k-Y^{k+1}\|_F^2}\nonumber\\
		&&\le \frac{\beta}{2}\|X^{k}-X^*\|_F^2+\frac{\beta}{2}\|Y^{k}-Y^*\|_F^2+\frac{1}{2\beta}\|Z^{k}-Z^*\|_F^2+\frac{\beta}{3-\sqrt{5}}\|Y^{k-1}-Y^{k}\|_F^2\nonumber\\
		&&-\frac{\varepsilon}{2}\beta\|X^k-X^{k+1}\|_F^2-(\kappa-\frac{\sqrt{5}+2}{2}\beta)\|Y^{k+1}-Y^k\|_F^2-\frac{1}{2\beta}\|\Lambda^k-\Lambda^{k+1}\|_F^2\nonumber\\
		&&+(\kappa-\frac{1}{2(\sqrt{5}-2-\varepsilon)\beta})\|Y^{k+1}-Y^*\|_F^2.
	\end{eqnarray}
	%where
	%\begin{eqnarray}\label{DEL}
	%&&\Delta_{k+1}:=\frac{\epsilon_1}{2\beta}\|\Lambda^k-\Lambda^{k+1}\|_F^2+\epsilon_2\|Y^k-Y^{k+1}\|_F^2\nonumber\\
	%&&+\left(\frac{\gamma+1}{2}\beta-\frac{1}{2(1-\epsilon_1)}\beta-
	%\frac{\alpha}{4(\frac{\beta}{2}+\kappa-\frac{\alpha}{4(\kappa-\alpha)}-\epsilon_2)}-{\color{red}\frac{\beta^2}{4(\kappa-\alpha)}}\right)\|X^k-X^{k+1}\|_F^2,
	%\end{eqnarray}
	%and the matrix
	%
	%\[   \label{G}
	% G=
	%   \left(\begin{array}{ccc}
	%
	%    (\gamma+1) \beta I_m  &     0      &   0  \\
	%     0    &    \beta I_m     &   0  \\
	%      0      &      0     & {1\over \beta} I_m
	% \end{array}\right),
	%\]
	%and the positive scalars $\epsilon_{1,2}$ are sufficient small.
\end{theorem}
\begin{proof}First, invoking (\ref{MSVI-F}) and $X^*+Y^*+Z^*-{\bar M}=0$, we have
	\begin{eqnarray}\label{WYX}
		&&\theta(U^{k+1}) - \theta(U^*)+  \langle W^{k+1}-W^*, \Psi(W^{k+1})\rangle+\beta\langle X^*+Y^*+Z^*-{\bar M}, \Gamma(X^k,Y^k,Z^k)\rangle \nonumber\\
		&&\quad\quad\le-\frac{\kappa}{2}\|Y^{k+1}-Y^*\|_F^2.
	\end{eqnarray}
	Then, setting $W:=W^*\in{\cal W}^*$ in (\ref{solstr}) and combining with (\ref{WYX}), we obtain that
	\begin{eqnarray}\label{contrin1}
		\lefteqn{0\ge
			\frac{\beta}{2}(\|V^{k+1}-V^*\|_Q^2+\|V^k-V^{k+1}\|^2_Q-\|V^k-V^*\|_Q^2)+\kappa\|Y^{k+1}-Y^k\|_F^2+\kappa\|Y^{k+1}-Y^*\|_F^2}\nonumber\\
		&&-\beta\langle X^{k+1}-X^k,Y^{k+1}-Y^k-(Y^k-Y^{k-1}) +\beta \langle Y^{k+1}-Y, X^{k+1}-X^k \rangle.
	\end{eqnarray}
	Next, adding  (\ref{Cau1})-(\ref{Cau3}) together, then substituting the resulting inequality into (\ref{contrin1}), we derive the assertion (\ref{Contr}) directly.\qed
\end{proof}

Based on the above theorem, we have the following theorem
immediately.

\begin{theorem}\label{theorem34}
	When $\beta$ is restricted by
	\begin{eqnarray}\label{beta}
		\beta\in\left(0,\;2(\sqrt{5}-2)\kappa \right),
	\end{eqnarray} there exists a sufficient small scalar $\varepsilon>0$ such that
	\begin{align}
		&\label{cof1} \kappa-\frac{\sqrt{5}+2}{2}\beta>0,\;\mbox{and}\;\;\kappa-\frac{1}{2(\sqrt{5}-2-\varepsilon)\beta}>0.
	\end{align}
	Then, we have
	\begin{enumerate}
		\item[(1)] The sequence $\{V^k\}$ is bounded.
		
		\item[(2)] $ \lim_{k\rightarrow \infty}\{\|Y^k - Y^{k+1} \|_F^2 +
		\|X^k - X^{k+1} \|_F^2 + \| \Lambda^k -{\Lambda}^{k+1}\|_F^2\}=0.$
	\end{enumerate}
\end{theorem}
\noindent{\it Proof}.  The inequality (\ref{cof1}) is elementary.
Note that the assertion (1)  follows from (\ref{Contr}) directly. Furthermore, we get
\begin{eqnarray*}
	\lefteqn{  \sum_{k=1}^{\infty} \left[\frac{\varepsilon}{2}\beta\|X^k-X^{k+1}\|_F^2+(\kappa-\frac{\sqrt{5}+2}{2}\beta)\|Y^{k+1}-Y^k\|_F^2+\frac{1}{2\beta}\|\Lambda^k-\Lambda^{k+1}\|_F^2
		\right]} \nonumber\\
	&&   \le
	\frac{\beta}{2}\|X^{1}-X^*\|_F^2+\frac{\beta}{2}\|Y^{1}-Y^*\|_F^2+\frac{1}{2\beta}\|Z^{1}-Z^*\|_F^2+\frac{\beta}{3-\sqrt{5}}\|Y^{0}-Y^{1}\|_F^2
	< +\infty,
\end{eqnarray*}
which immediately implies that
\begin{equation}
\label{THME2-1}
\lim_{k\to\infty} \|Y^k-Y^{k+1}\|_F =
0,\quad\lim_{k\to\infty}\|X^k-X^{k+1}\|_F=0,\quad \lim_{k\to\infty}
\|\Lambda^k-{ \Lambda}^{k+1}\|_{F}=0,
\end{equation}
i.e., the second assertion. \qquad $\Box$

We are now ready to prove the convergence of UMA.

\begin{theorem}\label{THMD2} Let $\{V^k\}$ and $\{W^{k}\}$
	be the sequences generated by UMA. Assume that  the penalty parameter $\beta$ is satisfied with (\ref{beta}).
	Then, we have
	\begin{enumerate}
		
		\item Any cluster point of  $\{{W}^{k}\}$ is a solution point of (\ref{MA-VIxyl}).
		
		\item The sequence $\{{V}^k\}$ converges to some $V^{\infty}\in {\cal
			V}^*$.
		\item The sequence $\{U^k\}$ converges to a solution point of (\ref{CouplePCA}).
	\end{enumerate}
\end{theorem}
\begin{proof} Since $\{W^k\}$ is bounded due to (\ref{Contr}), it has at least one cluster point. Let $W^{\infty}$ be a cluster point of
	$\{W^{k}\}$ and the subsequence $\{W^{k_j}\}$ converges to $W^{\infty}$.
	Because of the assertion (\ref{THME2-1}), it
	follows from (\ref{optwhol}) that
	\[ \nonumber
	\begin{array}{ll}
	\left\{
	\begin{array}{l}
	\langle Z- Z^{\infty},  - \Lambda^{\infty}\rangle \ge 0, \\
	\|X\|_*-\|X^{\infty}\|_* + \langle X- X^{\infty}, -\Lambda^{\infty}\rangle \ge 0, \\
	\tau\|Y\|_1-\tau\|Y^{\infty}\|_1 + \langle Y- Y^{\infty}, -\alpha M +\kappa Y^{\infty}-\Lambda^{\infty}\rangle \ge 0, \\
	\langle \Lambda -\Lambda^{\infty}, X^{\infty}+Y^{\infty}+Z^{\infty}-{\bar M}\rangle\ge0,
	\end{array} \right.
	& \;\; \forall\; W=(Z^\top,X^\top,Y^\top,\Lambda^\top)^\top\in  {\cal W}.
	\end{array}
	\]
	Thus,
	$$\theta(U) - \theta(U^{\infty})+  (W-W^{\infty})^\top \Psi(W^{\infty})\ge\frac{\kappa}{2}\|Y-Y^\infty\|_F^2,  \;\; \forall\; W=(Z^\top,X^\top,Y^\top,\Lambda^\top)^\top\in  {\cal W}.$$
	This means that $W^{\infty}$ is a solution of $\hbox{VI}({\cal W},\Psi,\theta) $.
	Then the inequality (\ref{Contr})
	is also valid if $V^{*}$ is replaced by $V^\infty$.
	Therefore, the non-increasing sequence $\{\frac{1}{2}\|V^k-V^\infty\|_Q^2+\frac{\beta}{3-\sqrt{5}}\|Y^k-Y^{k+1}\|_F^2\}$ converges to 0 since
	it has a subsequence $\{\frac{1}{2}\|V^{k_j}-V^\infty\|_Q^2+\frac{\beta}{3-\sqrt{5}}\|Y^{k_j}-Y^{k_j+1}\|_F^2\}$ converges to 0.
	Thus, the sequence $\{{V}^k\}$ converges to some $V^{\infty}\in {\cal
		V}^*$.
	Also, the updating scheme of $\Lambda^{k+1}$ in (\ref{GADM}) implies that
	$$ Z^{k+1}={\bar M}-X^{k+1}-Y^{k+1}+\frac{1}{\beta}(\Lambda^k-\Lambda^{k+1}).
	$$
	Combining the above equality, (\ref{THME2-1}) and $\lim_{k\rightarrow \infty}\|V^k-V^{\infty}\|_Q^2=0$,
	we have $W^k$ converges to $W^\infty$. It implies that the sequence $U^k$ converges to a solution point of  (\ref{CouplePCA}). Thus, the third assertion holds.
	%The results follow from Theorem 3.1 in \cite{tao2016convergence} directly.
\end{proof}
\begin{remark}
	Note that the range for $\beta$ in (\cite{tao2016convergence}) with convergence guarantee   is
	$(0,0.4\kappa)$ for UMA solving (\ref{CouplePCA}).
	However, we get a much larger range for the penalty parameter $\beta$ in (\ref{beta}).
\end{remark}
Next, we present a worst-case $O(1/t)$ convergence rate measured by the iteration complexity for UMA.
Indeed, the range of $\beta$ to ensure the $O(1/t)$ convergence rate is slightly more restrictive than (\ref{beta}).
Let us define
\begin{eqnarray*}
	{ Z}^{k+1}_t=\frac{1}{t}\sum_{k=1}^t Z^{k+1},\; { X}_t^{k+1}=\frac{1}{t}\sum_{k=1}^t X^{k+1},\;{Y}_t^{k+1}=\frac{1}{t}\sum_{k=1}^t Y^{k+1},\;\mbox{and}\; {U}^{k+1}_t=\frac{1}{t}\sum_{k=1}^t U^{k+1},\;{ W}^{k+1}_t=\frac{1}{t}\sum_{k=1}^t W^{k+1}.
\end{eqnarray*}
Obviously, ${ W}^{k+1}_t\in{\cal W}$ because of the convexity ${\cal W}$.
By invoking Theorem \ref{theorem34}, there exists a constant $C$ such that
\begin{eqnarray*}
	\max\left( \|X^k\|_F,\|Y^k\|_F,\|Z^k\|_F,\|\Lambda^k\|_F\right)\le C,\,\;\forall \;k.
\end{eqnarray*}
Next, we present several lemmas to facilitate the convergence rate analysis.

\begin{lemma}\label{ergo}
	Let $\{W^k\}$ be generated by UMA.
	Suppose that $0<\varepsilon <\sqrt{33}-5$.
	Then, it holds that
	\begin{align}
		&-\beta\langle X^{k+1}- X^{k},Y^{k+1}-Y^{k} \rangle \ge \beta\left(-\frac{7-\sqrt{33}}{8}\|X^k-X^{k+1}\|_F^2 -\frac{7+\sqrt{33}}{8}\|Y^{k+1}-Y^k\|_F^2\right),\label{Cau1E}\\
		&\beta\langle X^{k+1}-X^k, (Y^k-Y^{k-1})\rangle\ge \beta\left(-\frac{7-\sqrt{33}}{8}\|X^k-X^{k+1}\|_F^2  -\frac{7+\sqrt{33}}{8} \|Y^k-Y^{k-1}\|_F^2 \right),\label{Cau2E}\\
		&\beta \langle Y^{k+1}-Y, X^{k+1}-X^k \rangle\ge- \beta \left( \frac{1}{\sqrt{33}-5-\varepsilon}\|Y^{k+1}-Y\|_F^2+\frac{\sqrt{33}-5-\varepsilon}{4}\|X^{k+1}-X^k\|_F^2\right).\label{Cau3E}
	\end{align}
\end{lemma}
\begin{proof}
	These three inequalities follow from Cauchy-Schwarz inequality.\qed
\end{proof}
\medskip

\begin{lemma}\label{lem41}
	Let $\{W^k\}$ be the sequence generated by UMA (\ref{GADM}). If $\beta$ is restricted by
	\begin{eqnarray}\label{betacomp}
		\beta\in\left(0, \frac{\sqrt{33}-5}{2}\kappa\right),
	\end{eqnarray}
	then we have
	\begin{eqnarray}\label{Erogcont}
		\Theta(V^{k+1}, V^k,V)\le \Theta(V^{k}, V^{k-1},V)+\Xi(W^{k+1},W^k,W),
	\end{eqnarray}
	where
	\begin{align}
		&\Theta(V^{k+1}, V^k,V):=\frac{1}{2}\|V^{k+1}-V\|_Q^2+\frac{7+ \sqrt{33}}{8}\beta\|Y^{k+1}-Y^{k}\|_F^2. \;\label{The}
	\end{align}
	and
	\begin{eqnarray}\Xi(W^{k+1},W^k,W):=\theta(U)-\theta(U^{k+1})+(W-W^{k+1})^\top \Psi (W)+\beta\langle{\cal R}, Y^k-Y^{k+1}+X^k-X^{k+1}\rangle. \label{LXi}
	\end{eqnarray}
\end{lemma}
\begin{proof}
	First, summing inequalities (\ref{Cau1E})-(\ref{Cau3E}) together, we get
	\begin{eqnarray*}
		\lefteqn{\beta\langle X^{k+1}- X^{k},Y^{k+1}-Y^{k} -(Y^k-Y^{k-1})\rangle +\beta \langle Y^{k+1}-Y, X^{k+1}-X^k \rangle}\nonumber\\
		&\ge& \frac{\varepsilon-2}{4}\beta\|X^{k+1}-X^k\|_F^2-\frac{7+\sqrt{33}}{8}\beta\|Y^{k+1}-Y^k\|_F^2\nonumber\\
		&&-\frac{7+\sqrt{33}}{8}\beta\|Y^k-Y^{k-1}\|_F^2-\frac{1}{\sqrt{33}-5-\varepsilon}\beta\|Y^{k+1}-Y\|_F^2.
	\end{eqnarray*}
	Then, substituting  the above inequality into (\ref{solstr}) and invoking (\ref{The}), (\ref{LXi}), we obtain
	\begin{eqnarray*}
		\Theta(V^{k+1}, V^k,V)&\le& \Theta(V^{k}, V^{k-1},V)+\Xi(W^{k+1},W^k,W)-X^k\|_F^2-(\kappa-\frac{5+\sqrt{33}}{4}\beta)\|Y^{k+1}-Y^k\|_F^2\nonumber\\
		&&	-\frac{\beta}{4}\varepsilon\|X^{k+1}-\frac{1}{2\beta}\|\Lambda^k-\Lambda^{k+1}\|_F^2-
		(\frac{\kappa}{2}-\frac{1}{\sqrt{33}-5-\varepsilon}\beta)\|Y^{k+1}-Y\|_F^2.
	\end{eqnarray*}
	Let $\varepsilon\rightarrow 0+$, the assertion follows directly.
\end{proof}
\medskip

\begin{theorem}\label{TA2}
	For $t$ iterations generated by UMA  with $\beta$ restricted in
	the following assertions holds.
	\begin{itemize}
		\item[(1)] We have
		\begin{eqnarray}\label{BK}
			\lefteqn{\theta(U_t^{k+1})-\theta(U)+(W_t^{k+1}-W)^\top \Psi(W)}\nonumber\\
			&\le& \frac{1}{t}\left[4\beta C\|X+Y+Z-{\bar M}\|_F
			+\frac{1}{2}\|V^1-V\|_Q^2+\frac{7+\sqrt{33}}{8}\beta\|Y^1-Y^0\|_F^2\right].
		\end{eqnarray}
		
		\item[(2)] There exists a constant ${\bar c}_1>0$ such that
		\begin{eqnarray}\label{primalergo}
			\|{ X}^{k+1}_t+{ Y}^{k+1}_t+{ Z}^{k+1}_t-{\bar M}\|^2\le \frac{{\bar c}_1}{t^2}.
		\end{eqnarray}
		%\begin{eqnarray}\label{ergoprimal}
		%\end{eqnarray}
		\item[(3)] There exists a constant ${\bar c}_2>0$ such that
		\begin{eqnarray}\label{objergo}
			|\theta(U_t^{k+1})-\theta(U^*)|\le\frac{{\bar c}_2}{t}.
		\end{eqnarray}
	\end{itemize}
\end{theorem}
\begin{proof}
	1) First,
	it follows from  the assertion \eqref{Erogcont} that
	for all $W\in{\cal W}$, we have
	\begin{eqnarray}\label{Lem2p}
		&&\theta(U)-\theta(U^{k+1})+(W-W^{k+1})^\top \Psi(W)+\beta\langle {\cal R}, Y^k-Y^{k+1}+X^k-X^{k+1}\rangle \nonumber\\
		&&\quad\quad\ge \Theta(V^{k+1}, V^k,V )
		-\Theta(V^{k}, V^{k-1},V ).
	\end{eqnarray}
	Summarizing both sides of the above inequalities from $k=1, 2, \cdots, t$, we have
	\begin{eqnarray}\label{sum}
		&&t\theta(U)-\sum_{k=1}^t\theta(U^{k+1})+(tW-\sum_{k=1}^t W^{k+1})^\top \Psi(W)+\beta\langle {\cal R}, Y^1-Y^{t+1}+X^1-X^{t+1}\rangle\nonumber\\
		&&\quad\quad\ge \Theta(V^{t+1}, V^t, V )-\Theta(V^{1}, V^0,V ).
	\end{eqnarray}
	Then, it follows from the convexity of $\theta$  that
	\begin{equation}
	\label{FK}
	\theta(U_t^{k+1})\le\frac{1}{t}\sum_{k=1}^t\theta(U^{k+1}).
	\end{equation}
	Combining  (\ref{sum}) and (\ref{FK}), we have
	\begin{eqnarray}\label{add1}
		\theta(U_t^{k+1})-\theta(U)+( W_t^{k+1}-W)^\top \Psi(W)\le\frac{1}{t} \left(\Theta(V^{1}, V^0,V )+4\beta C\|{\cal R}\|_F\right).
	\end{eqnarray}
	Thus, the assertion (\ref{BK}) follows from the above inequality and the defintion of $\Theta(V^{1}, V^0,V)$ directly.
	
	2) Let us define ${\bar c}_1=\frac{2}{\beta^2}\left(\|\Lambda^1-\Lambda^*\|^2+\|\Lambda^{k+1}-\Lambda^*\|^2\right)$. Then, we have
	\begin{eqnarray}\nonumber
		&&\|{ X}^{k+1}_t+{ Y}^{k+1}_t+{ Z}^{k+1}_t-{\bar M}\|^2 =\left\|\frac{1}{t}\sum_{k=1}^t\left[ { X}^{k+1}+{ Y}^{k+1}+{ Z}^{k+1}-{\bar M}\right]\right\|^2 =\left\|\frac{1}{t}\sum_{k=1}^t\left[  \frac{1}{\beta}(\Lambda^k-\Lambda^{k+1})\right]\right\|^2\nonumber\\
		&&\quad\quad=\left\| \frac{1}{t\beta}\left( \Lambda^1-\Lambda^{t+1}\right)\right\|^2\le\frac{2}{t^2\beta^2}\left(\|\Lambda^1-\Lambda^*\|^2+\|\Lambda^{k+1}-\Lambda^*\|^2\right)=\frac{{\bar c_1}}{t^2},\nonumber
	\end{eqnarray}
	%{\color{blue}where the first equality follows from (\ref{XK}), the second follows from (\ref{lap}) and the last follows from Lemma \ref{lem211}.}
	The assertion (\ref{primalergo}) is proved immediately.
	
	3) It follows from ${\cal L}(U_t^{k+1},\Lambda^*)\ge {\cal L}(U^*,\Lambda^*)$ with ${\cal L}$ defined in (\ref{AL-Fun}) that
	\begin{eqnarray}
		\label{objlowerg}
		\lefteqn{\theta(U_t^{k+1})-\theta(U^*)\ge\langle \Lambda^*,X_t^{k+1}+Y_t^{k+1}+Z_t^{k+1}-{\bar M}\rangle}\nonumber\\
		&&\ge-\frac{1}{2}\left(\frac{1}{t}\|\Lambda^*\|^2+t\|X_t^{k+1}+Y_t^{k+1}+Z_t^{k+1}-{\bar M}\|^2 \right)\ge-\frac{1}{2t}(\|\Lambda^*\|^2+{\bar c}_1),
	\end{eqnarray}
	where the second inequality is due to Cauchy-Schwarz inequality, and the last is due to (\ref{primalergo}).
	On the other hand, setting $W:=W^*$ in (\ref{add1}), we obtain
	\begin{eqnarray}\nonumber
		%\label{BK2}
		&& \theta(U_t^{k+1})-\theta(U^*)+\langle W_t^{k+1}-W^*,\Psi(W^*)\rangle\le\frac{1}{t}\Theta(V^{1}, V^0,V^*).
	\end{eqnarray}
	Invoking the definition of $\Psi$ in (\ref{MD-F-2}), we have
	\begin{eqnarray}\nonumber
		%\label{correrg}
		( W_t^{k+1}-W^*)^\top \Psi(W^*)=-\langle \Lambda^*,{ X}^{k+1}_t+{ Y}^{k+1}_t+{ Z}^{k+1}_t-{\bar M}\rangle\ge-\frac{1}{2t}(\|\Lambda^*\|^2+{\bar c}_1),
	\end{eqnarray}
	where the proof of the last inequality is similar to (\ref{objlowerg}).
	Combining these two inequalities above, we get
	\begin{eqnarray}\label{objuperg}
		\theta(U_t^{k+1})-\theta(U^*)\le\frac{1}{t}\Theta(V^{1}, V^0,V^*)+\frac{1}{2t}(\|\Lambda^*\|^2+{\bar c}_1).
	\end{eqnarray}
	The inequalities (\ref{objlowerg}) and (\ref{objuperg}) indicate that the assertion (\ref{objergo}) holds by setting $\bar c_2:=\Theta(V^{1}, V^0,V^*)+\frac{1}{2}(\|\Lambda^*\|^2+{\bar c}_1)$.
\end{proof}

£ý

\end{document}